\documentclass[11pt]{article}

\usepackage{tikz}
\usepackage{hyperref}
\usepackage{amssymb,amsmath,amsfonts,fullpage,mdwlist,enumerate,amsthm, graphicx,algpseudocode, algorithm, algorithmicx, cite, xspace, thm-restate, framed}
\usepackage{cleveref}

\newcommand{\remove}[1]{}

\newtheorem{theorem}{Theorem}
\newtheorem{claim}[theorem]{Claim}
\newtheorem{lemma}[theorem]{Lemma}

\newtheorem{definition}[theorem]{Definition}
\newtheorem{question}[theorem]{Question}
\newcommand{\poly}{\mathop\mathrm{poly}}
\newcommand{\clique}{\textsc{Congested Clique}\xspace}
\newcommand{\local}{\textsc{Local}\xspace}
\newcommand{\knearest}{$k$\textsc{-nearest}\xspace}
\newcommand{\kdnearest}{$(k,d)$\textsc{-nearest}\xspace}
\newcommand{\ktnearest}{$(k,t)$\textsc{-nearest}\xspace}
\newcommand{\sdk}{$(S, d)$\textsc{-source detection}\xspace}
\newcommand{\distthrough}{\textsc{distance through sets}\xspace}
\newcommand{\BB}{\{ 0,1 \}}
\newcommand{\PRG}{\mathsf{PRG}}
\newcommand{\Exp}[2]{\mathop{\mathbf E}_{#1}\! \left[ {#2} \right]}
\newcommand{\ith}[1]{{#1}\textsuperscript{th}}

\newcommand{\cH}{{\mathcal H}}
\newcommand{\prob}[2]{\Pr_{#1}\! \left[ #2 \right]}
\newcommand{\pr}[1]{\Pr\! \left[ {#1} \right]}
\newcommand{\E}[1]{\mathop{\mathbf E}\! \left[ {#1} \right]}
\def\cA{{\cal A}}
\def\cC{{\cal C}}

\begin{document}
%\title{Exponentially Faster Congested Clique Algorithms \\for Approximate Shortest Paths}
\title{Exponentially Faster Shortest Paths \\ in the Congested Clique}
\author{
 Michal Dory\\
  \small Technion\\
  \small smichald@cs.technion.ac.il 
\and
Merav Parter\\
        \small Weizmann Institute \\
        \small merav.parter@weizmann.ac.il
}

\maketitle

\begin{abstract}
We present improved deterministic algorithms for approximating shortest paths in the \clique model of distributed computing. We obtain $\poly(\log\log n)$-round algorithms for the 
following problems in unweighted undirected $n$-vertex graphs:
\begin{itemize}
\item $(1+\epsilon)$-approximation of multi-source shortest paths (MSSP) from $O(\sqrt{n})$ sources.
\item $(2+\epsilon)$-approximation of all pairs shortest paths (APSP).
\item $(1+\epsilon,\beta)$-approximation of APSP where $\beta=O(\frac{\log\log n}{\epsilon})^{\log\log n}$. 
\end{itemize}
These bounds improve \emph{exponentially} over the state-of-the-art poly-logarithmic bounds due to [Censor-Hillel et al., PODC19]. 
It also provides the first nearly-additive bounds for the APSP problem in sub-polynomial time.  Our approach is based on distinguishing between short and long distances based on some distance threshold $t = O(\frac{\beta}{\epsilon})$ where $\beta=O(\frac{\log\log n}{\epsilon})^{\log\log n}$. Handling the long distances is done by devising a new algorithm for computing sparse $(1+\epsilon,\beta)$ emulator with $O(n\log\log n)$ edges. For the short distances, we provide \emph{distance-sensitive} variants for the distance tool-kit of [Censor-Hillel et al., PODC19]. By exploiting the fact that this tool-kit should be applied only on local balls of radius $t$, their round complexities get improved from $\poly(\log n)$ to $\poly(\log t)$. 

Finally, our deterministic solutions for these problems are based on a derandomization scheme of a novel variant of the hitting set problem, which might be of independent interest. %(e.g., reducing the size of existing spanner constructions). 
\end{abstract}

\newpage 
\tableofcontents
\newpage
\section{Introduction} 
Shortest paths computation is one of the most fundamental graph problems, and as such it has been studied extensively in various computational models. In this work we consider the problem of approximating shortest path distances in the \clique model of distributed computing \cite{lotker2005minimum}.  %\mtodo{I moved the description of the model to the preliminaries to add more details about the output.}
This model attracts an increasing amount of attention over the last decade, due to its relation to practical parallel models such as MapReduce \cite{hegeman2015lessons} and $k$-machines \cite{klauck2014distributed}.
%The long line of distance computation works in the \clique model 
The fast evolution of shortest path algorithms in the \clique model 
\cite{le2016further,holzer2016approximation,censor2019algebraic,becker2017near,ParterY18,ChechikM19,elkin2019linear,censor2019sparse,dinitz2019brief,elkin2019hopsets, DBLP:conf/podc/Censor-HillelDK19} is characterized by two main mile-stones.
%
%
%
%
%This work has led to the first $o(\sqrt{n})$ algorithms for the fundamental problems of SSSP (single-source shortest paths) and $(2+\epsilon)$-approximation for the APSP (all-pairs-shortest-paths).
%Since then, there has been a long line of works that studied distance computation in the \clique model 
%\cite{le2016further,holzer2016approximation,censor2019algebraic,becker2017near,ParterY18,ChechikM19,elkin2019linear,censor2019sparse,dinitz2019brief,elkin2019hopsets, DBLP:conf/podc/Censor-HillelDK19}. Most notably is the recent work by Censor-Hillel et al. that presented poly-logarithmic algorithms for the approximates MSSP and APSP problems.

%Many of the state-of-the art results exploit the connection to matrix multiplication 
%\cite{censor2019sparse,censor2019algebraic,le2016further}. Approximating distances can also be obtained by computing a sparse $k$-spanner and having all nodes learn the entire spanner. 
\vspace{-8pt}
\paragraph{First Era: Polynomial Complexity for Exact Distances.} 
In their influential work \cite{censor2019algebraic}, Censor-Hillel et al. presented the first efficient algorithm for fast matrix multiplication in the \clique model, and demonstrated its applicability for exact distance computation. 
%
%The first non-trivial distance computation algorithms in the \clique model were mostly based on fast matrix multiplication. 
%Censor et al. presented several polynomial-time algorithm for computing exact distances via matrix multiplication. 
Specifically, they gave an $O(n^{0.158})$-round algorithm for computing all-pairs shortest paths (APSP) in unweighted undirected APSP. For weighted and directed graphs, they showed an $O(n^{0.158})$-round algorithm for $(1+o(1))$-approximate APSP, and an $\widetilde{O}(n^{1/3})$-round algorithm for exact APSP. In a follow-up work, Le Gall \cite{le2016further}
extended this algebraic framework, and improved the complexity of exact APSP to $O(n^{0.2096})$ rounds for weighted directed graphs with constant weights. Izumi and Le Gall showed that using quantum computing, the complexity of APSP can be improved to $\tilde{O}(n^{1/4})$ for weighted directed graphs \cite{izumi2019quantum}.
Censor-Hillel, Leitersdorf and Turner \cite{censor2019sparse} showed improved algorithms for \emph{sparse} matrix multiplication. In particular, their algorithm multiplies two $n\times n$ matrices, each with $O(n^{3/2})$ non-zero entries, within constant number of rounds. This led to an exact APSP computation in $O(D \cdot (m/n )^{1/3})$ rounds where $D$ is the (unweighted) diameter of the $n$-vertex graph and $m$ is the number of edges. 
%In sum, the early algorithms for distance computation in this model were obtained by reduction to matrix multiplication.
\vspace{-8pt}
\paragraph{Second Era: Poly-logarithmic Complexity for Approximate Distances.}
The above mentioned algorithms are all based on matrix multiplication. 
An alternative natural approach for approximating APSP is obtained by the notion of multiplicative spanners, namely, sparse subgraphs that preserve all distances up to a small multiplicative stretch. 
Specifically, by computing a sparse $k$-spanner in the graph, and collecting its edges at each vertex, one immediately gets
a $k$-approximation for the APSP problem. Note, however, that due to the size vs. stretch tradeoff of multiplicative spanners, in poly-logarithmic time, the approximation value $k$ is at least logarithmic.
The problem of obtaining $O(1)$ distance approximation in poly-logarithmic time was left widely open until the recent works of \cite{becker2017near} and \cite{censor2019sparse}. 

By using the ingenious method of \emph{continuous optimization} \cite{sherman2013nearly}, Becker et al. \cite{becker2017near} presented the first poly-logarithmic algorithm for the $(1+\epsilon)$-approximate single-source shortest paths (SSSP) problem in weighted undirected graphs. The starting point of their algorithms is based on the $O(\log n)$-approximation provided by spanners. This approximation is then iteratively improved by applying a tailored gradient descent procedure. Their algorithms also have the additional benefit of being implemented in the more stringent Broadcast \clique model\footnote{In this model, each vertex is limited to send the \emph{same} message to all other vertices in a given round.}.

In a very recent work, Censor-Hillel et al. \cite{DBLP:conf/podc/Censor-HillelDK19} presented  poly-logarithmic time solutions for the more general problems of approximate multi-source shortest paths (MSSP) and APSP. In particular, among other results, they showed an $O(\log^2 n/\epsilon)$-round algorithm for the $(2+\epsilon)$-approximate APSP problem in \emph{unweighted} undirected graphs, an $O(\log^2 n/\epsilon)$-round algorithm for $(3+\epsilon)$-approximate APSP in \emph{weighted} undirected graphs; as well as $O(\log^2 n/\epsilon)$-round algorithm for the MSSP problem in weighted undirected graphs with respect to $O(\sqrt{n})$ sources. The approach of \cite{DBLP:conf/podc/Censor-HillelDK19} is based on a clever reduction to sparse matrix multiplication algorithms from \cite{censor2019sparse, DBLP:conf/podc/Censor-HillelDK19}. This reduction is not immediate, for the following reason. The basic approach of using matrix multiplications to compute distances is based on starting with the input matrix of the graph $A$ and then iteratively square it $O(\log{n})$ times.  
Even if the original graph is sparse (e.g., with $O(n^{3/2})$ edges), a single multiplication step might already result in a dense matrix. Since one needs to apply $O(\log n)$ multiplication steps, it is required to sparsify the intermediate matrices as well. To overcome this barrier, \cite{DBLP:conf/podc/Censor-HillelDK19} developed a rich tool-kit for distance computation, that includes hopset constructions, fast computation of the $k$-nearest neighbours and more.

%carefully designed to exploit the sparse matrix multiplication algorithm of \cite{censor2019sparse}. 
%The poly-logarithmic time complexity appears to be quite essential in the matrix multiplications based algorithms, Intuitively, in iteration $i$, vertices learn about paths that have at most $2^i$ edges, hence to learn about 
The poly-logarithmic time complexity appears to be a natural barrier in the matrix multiplications based algorithms, from the following reason. Basically these algorithms work in iterations, where in iteration $i$, vertices learn about paths that have at most $2^i$ edges, hence $\Omega(\log n)$ iterations seem to be required in order to capture shortest paths with polynomial number of hops. 
In this paper we ask if this barrier is indeed fundamental. 
%This raises the following question that underlies the key motivation for our paper:
%\vspace{-3pt}
\begin{question}
Is it possible to provide $(1+\epsilon)$-approximation for SSSP and $(2+\epsilon)$-approximation for APSP in $o(\log n)$ rounds?
\end{question}
%\vspace{-3pt}
Another plausible explanation for viewing the logarithmic time complexity as a natural barrier comes from the following observation. 
The current state-of-the-art bound of many graph problems in the \clique model is bounded from below by the 
logarithm of their respective \local time complexities\footnote{In the \local model, only vertices that are neighbors in the input graph can communicate in a given round. The messages are allowed to be of unbounded size.}. In the more restrictive, yet quite related model of MPC, this limitation was conditionally proven by Ghaffari, Kuhn and Uitto \cite{ghaffari2019conditional}. Since the locality of the shortest path problem is linear, we reach again to the barrier of $O(\log n)$ rounds.

%Our work is also motivated by the lack of any sub-polynomial algorithms for nearly-additive APSP in unweighted graphs. As noted in \cite{DBLP:conf/podc/Censor-HillelDK19}, the works of \cite{dor2000all,korhonen2018towards} imply that one cannot compute a $(2-\epsilon)$ approximate APSP in poly-logarithmic time, due to a reduction to fast matrix multiplication. This, however, does not rule out the possibility of obtaining nearly-additive approximation.
Another interesting question left by \cite{DBLP:conf/podc/Censor-HillelDK19} is whether is it possible to design efficient algorithms for \emph{additive} or \emph{near-additive} approximation for APSP. Here, the goal is to get an estimate $\delta(u,v)$ for the distance between $u$ and $v$ such that $\delta(u,v) \leq (1+\epsilon) d(u,v) + \beta$ for some small $\beta$. Note that such guarantee is closer to $(1+\epsilon)$-approximation for long enough paths, which is much better than a multiplicative $(2+\epsilon)$-approximation, which raises the following question. %\mtodo{slightly edited the text here.} 
 
%\vspace{-3pt}
\begin{question}
Can we get near-additive approximation for APSP in sub-polynomial time? 
\end{question} 

%\mtodo{moved this text here}
We answer both questions in the affirmative by devising a unified tool for computing nearly-additive approximate distances in $\poly(\log\log n)$ rounds. This immediately solves Question 2. To address Question 1,  
our approach is based on a mixture of the two current approaches for distance computation, namely, graph sparsification and matrix multiplications.
We start by describing our results, and then elaborate on our approach.
%
%\mptodo{This text might be edited or removed.} \mtodo{I think it can be removed from here, moved the discussion about emulators to the beginning of technical overview}
%We address both questions by designing a nearly-additive emulator, namely, a sparse graph\footnote{Unlike spanners, the emulator is not necessarily a subgraph of the input graph. It might contain edges that are not in $G$, and it is allowed to be weighted, even when the graph $G$ is unweighted.} that preserves the distances up to a nearly-additive stretch. Collecting all the edges of this emulator already handles Question 2. In the next subsection, we further elaborate on our approach. 

%\vspace{-10pt}
%\subsection{Our Contributions and Approach in a Nutshell} \vspace{-3pt}
\subsection{Our Contributions}
In this work we break the natural logarithmic barrier for distance computation in the \clique model. We present $\poly(\log\log n)$ round approximation algorithms for the fundamental problems of MSSP and APSP. 
\vspace{-3pt}
\paragraph{$(1+\epsilon)$ MSSP.} Our first contribution is the first sub-logarithmic algorithm for computing distances from up to $O(\sqrt{n})$ sources.\footnote{The statements in this section are slightly simpler and assume that $\epsilon$ is constant just to simplify the presentation. The exact statements appear in Section \ref{sec:applications}.}

\begin{theorem}[$(1+\epsilon)$-MSSP in $\poly(\log\log n)$ Rounds]\label{thm:rand-MSSP}
For every $n$-vertex unweighted undirected graph $G=(V,E)$, a subset of $O(\sqrt{n})$ sources $S \subset V$ and a constant $\epsilon \in (0,1)$, there is a randomized algorithm in the \clique model that computes $(1+\epsilon)$-approximation for the $S \times V$ distances with high probability within $\widetilde{O}((\log\log n)^2)$ rounds.\footnote{The term $\widetilde{O}(x)$ hides factors that are poly-logarithmic in $x$.} 
\end{theorem}

We remark that even for computing $(1+\epsilon)$-approximate distances from a single source, the fastest previous algorithms require poly-logarithmic time \cite{DBLP:conf/podc/Censor-HillelDK19, becker2017near}.
Similarly to \cite{DBLP:conf/podc/Censor-HillelDK19}, the barrier of handling at most $O(\sqrt{n})$ sources is rooted at the sparse matrix multiplication algorithm of \cite{DBLP:conf/podc/Censor-HillelDK19, censor2019sparse}.
The latter takes $O(1)$ time as long as the density of the matrices is $O(\sqrt{n})$, for higher values the complexity is polynomial. 
%complexity of our algorithm depends on the number of sources, that affects the sparsity of problem. At a high-level, the reason our algorithm is efficient as long as the number of sources is $O(\sqrt{n})$ relies on the fact that the current fastest algorithms for sparse matrix application \cite{DBLP:conf/podc/Censor-HillelDK19, censor2019sparse} 
%
%\mptodo{Can we provide some intuition for why the MSSP gets stuck at $\sqrt{n}$? I.e., can we say that in order to support larger number of sources, we need faster matrix multiplication algorithms? Does the reduction from $(2-\epsilon)$ APSP to MM give any hint regarding the number of sources from which $(1+\epsilon)$ MSSP could not be obtain in poly-log time?}\mtodo{added a few lines.}
\vspace{-3pt}
\paragraph{$(2+\epsilon)$-APSP.}
Our next contribution is providing a $(2+\epsilon)$-approximation for unweighted APSP in $\poly(\log\log n)$ time, improving upon the previous $\poly(\log{n})$ algorithm \cite{DBLP:conf/podc/Censor-HillelDK19}.

\begin{theorem}[$(2+\epsilon)$-APSP in $\poly(\log\log n)$ Rounds]\label{thm:rand-APSP}
For every $n$-vertex unweighted undirected graph $G=(V,E)$, and constant $\epsilon \in (0,1)$, there is a randomized algorithm in the \clique model that computes $(2+\epsilon)$-approximation for APSP w.h.p within $\widetilde{O}((\log\log n)^2)$ rounds. 
\end{theorem}
%\vspace{-15pt} 
It is worth noting that a $(2+\epsilon)$-approximation is essentially the best we can hope for without improving the complexity of matrix multiplication in the \clique model.
As noted in \cite{DBLP:conf/podc/Censor-HillelDK19}, a reduction to matrix multiplication that appears in \cite{dor2000all,korhonen2018towards} implies that one cannot compute a $(2-\epsilon)$ approximate APSP in sub-polynomial time, without obtaining sub-polynomial algorithms for matrix multiplication.

\vspace{-3pt}
\paragraph{$(1+\epsilon,\beta)$-APSP.}

Finally, we provide $\poly(\log\log n)$-round algorithm for the $(1+\epsilon,\beta)$-approximation for APSP for $\beta = O(\frac{\log{\log{n}}}{\epsilon})^{\log{\log{n}}}$.

\begin{theorem}[Near-Additive APSP]\label{cor:near-additive-apsp}
For any constant $\epsilon \in (0,1)$, there exists a randomized algorithm in the \clique model that w.h.p. computes a $(1+\epsilon,\beta)$-approximation for the APSP problem in undirected unweighted graphs in $\widetilde{O}((\log\log n)^2)$ rounds, for $\beta = O(\frac{\log{\log{n}}}{\epsilon})^{\log{\log{n}}}$.
\end{theorem}

For the best of our knowledge, this is the first \emph{sub-polynomial} algorithm to the problem. A polynomial algorithm can be obtained either by the 
exact algorithm for unweighted APSP that takes $O(n^{0.158})$ time \cite{censor2019algebraic}, or by building a sparse near-additive spanner or emulator and letting the vertices learn it. This takes $O(n^{\rho})$ rounds where $\rho$ is an arbitrary small constant, using the constructions in \cite{elkin2019near, elkin2018efficient}. We remark that these constructions work also in the more general CONGEST model, where vertices can only send $\Theta(\log{n})$ bit messages to their neighbours.    

Note that while \cite{DBLP:conf/podc/Censor-HillelDK19} shows poly-logarithmic algorithms for multiplicative $(2+\epsilon)$-approximation of APSP, prior to this work there was no nearly-additive approximation in sub-polynomial time. This nearly-additive approximation yields a near-exact $(1+\epsilon)$-approximation for long distances (i.e., of distance $\Omega(\beta/\epsilon)$), improving upon the existing $(2+\epsilon)$-approximation.
\vspace{-5pt}

\paragraph{Derandomization.} The randomized constructions presented above share a similar randomized core for which we devise a new derandomization algorithm. The current deterministic constructions of spanners and hopsets are mostly based on the derandomization of the hitting set problem \cite{Censor-HillelPS16,ghaffari2018derandomizing,ParterY18}.  In our context, the standard hitting set based arguments lead to a logarithmic overhead in the size of the emulator, and consequently also in the number of rounds to collect it. This is clearly too costly as we aim towards sub-logarithmic bounds. 
To handle that, we isolate the key probabilistic argument that leads to the desired sparsity of our emulators. 
By using the strong tool of pseudorandom generators for fooling DNF formulas \cite{gopalan2012better}, and extending the approach of  \cite{ParterY18}, we show that this probabilistic argument can hold even when supplying all vertices a shared random seed of $\widetilde{O}(\log n)$ bits. 
As we believe that this derandomization scheme might be useful in other contexts, we formalize the setting by defining the \emph{soft hitting set} problem that captures our refined probabilistic arguments in a black-box manner, and show the following:
\begin{lemma}\label{lem:derand}
The soft hitting set problem can be solved in $O((\log\log n)^3)$ rounds. This leads to $\widetilde{O}((\log\log n)^4)$ deterministic solutions for Thm. \ref{thm:rand-MSSP}, \ref{thm:rand-APSP} and \ref{cor:near-additive-apsp}. 
\end{lemma}

%\mtodo{It would be good to discuss somewhere in the intro the following: 
%\begin{enumerate}
%\item Comparison to other algorithms for near-additive emulators, I guess the Elkin-Neiman construction is the most related, but also the Thorup-Zwick construction is somewhat related.
%\item Mentioning that getting a better than 2 approximation for APSP already solves matrix multiplication.
%\item Explain why $\poly(\log{n})$ seems like a natural barrier for global problems like APSP (for example: all the matrix multiplication based algorithms require at least $\log{n}$ iterations where iteration $i$ deals with paths with $2^i$ edges.)
%\item Mention that even for SSSP the best previous results require polylog time.
%%\item Add a technical overview, maybe with some toy example of an emulator with $n^{1+\frac{1}{4}}$ edges, and discuss through this the main challenges and ideas.
%\end{enumerate} }

%\mtodo{we should define the model either in the intro or preliminaries.}

\subsection{Our Techniques}

One of the main ingredients in our algorithms is a new fast algorithm for constructing a near-additive emulator with $O(n\log\log n)$ edges. Roughly speaking, the near-additive emulator is a sparse graph\footnote{Unlike spanners, the emulator is not necessarily a subgraph of the input graph. It might contain edges that are not in $G$, and it is allowed to be weighted, even when the graph $G$ is unweighted.} that preserves the distances up to a nearly-additive stretch. Collecting all the edges of this emulator at each vertex, already gives near-additive approximation for APSP. %In the next subsection, we further elaborate on our approach.
We start by explaining our emulator results and then explain how to use them for obtaining the approximate distances.

%We next explain the main ideas behind new fast algorithm for constructing a the near-additive emulator. 
%This emulator is a sparse graph that preserves the distances up to a $(1+\epsilon,\beta)$-stretch. 
%It turned out that this algorithm can be implemented very efficiently in the \clique model, in just $\poly(\log{\log{n}})$ rounds, which led to $\ poly(\log{\log{n}})$ algorithms not only for near-additive APSP, but also for a $(2+\epsilon)$-approximation for APSP, and a $(1+\epsilon)$-approximation for multi-source shortest paths, in \emph{unweighted} graphs.   We next explain how a near-additive emulator is useful for approximating shortest paths, and then describe the emulator construction.

%We address both questions by designing a nearly-additive emulator, namely, a sparse graph\footnote{Unlike spanners, the emulator is not necessarily a subgraph of the input graph. It might contain edges that are not in $G$, and it is allowed to be weighted, even when the graph $G$ is unweighted.} that preserves the distances up to a nearly-additive stretch. Collecting all the edges of this emulator already handles Question 2. In the next subsection, we further elaborate on our approach.

\vspace{-5pt}
\paragraph{Nearly Additive Emulators.} For an $n$-vertex unweighted graph $G=(V,E)$, a $(1+\epsilon,\beta)$ emulator $H=(V,E',w)$ is a sparse graph that preserves the distances of $G$ up to a $(1+\epsilon,\beta)$-stretch.
%\vspace{-10pt}

Nearly-additive emulators have been studied thoroughly mainly from the pure graph theoretic perspective \cite{thorup2006spanners,elkin2018efficient} . Our emulator algorithm is inspired by the two state-of-the-art constructions of Elkin and Neiman \cite{elkin2018efficient} and Thorup and Zwick \cite{thorup2006spanners}. 
The main benefit of our algorithm is its efficient implementation in the \clique model, which takes $O(\frac{\log^2{\beta}}{\epsilon})$ rounds w.h.p, for $\beta = O(\frac{\log{\log{n}}}{\epsilon})^{\log{\log{n}}}.$ For a constant $\epsilon$ this gives a complexity of $\widetilde{O}((\log\log n)^2)$ rounds.

\begin{lemma}[Near-Additive Sparse Emulators]\label{lem:near-additive-emulator}
For any $n$-vertex unweighted undirected graph $G=(V,E)$ and $0<\epsilon<1$, there is a randomized algorithm in the \clique model that computes $(1+\epsilon,\beta)$ emulator $H=(V,E',w)$ with $O(n\log\log n)$ edges within $O(\frac{\log^2{\beta}}{\epsilon})$ rounds w.h.p, where $\beta = O(\frac{\log{\log{n}}}{\epsilon})^{\log{\log{n}}}.$
\end{lemma}

We next explain how an emulator leads to fast algorithms for approximate shortest paths.

\subsubsection*{Approximating shortest paths via near-additive emulators}

%By collecting the edges of emulator at each vertex, we already get a near-additive approximation for APSP in $\widetilde{O}((\log\log n)^2)$ rounds, we next explain how an emulator helps in obtaining $\widetilde{O}((\log\log n)^2)$ algorithms also for $(2+\epsilon)$-approximation for APSP and $(1+\epsilon)$-approximation for MSSP. 

%The computation of $(1+\epsilon,\beta)$ emulators already handles long-distances. To take care of the short distances, we reconstruct the tool-kit of \cite{DBLP:conf/podc/Censor-HillelDK19}, and turn them into distance-sensitive primitives, i.e., with a running time that depends on a given distance threshold. 
%Intuitively, by computing the $(1+\epsilon,\beta)$ emulators, the locality of the remaining problem to be solved is $\beta/\epsilon$ rather than $n$. This in particular imply, for example, that the distance tools of \cite{DBLP:conf/podc/Censor-HillelDK19} should employ only $O(\log t)=\widetilde{O}((\log\log n)^2)$ steps of matrix multiplications, rather than $O(\log n)$. We have:

%Finally, we turn to consider the most general task of $(2+\epsilon)$-approximate APSP. Using the recipe of \cite{DBLP:conf/podc/Censor-HillelDK19}, we can easily convert the $(1+\epsilon)$-MSSP algorithm into a $(3+\epsilon)$ APSP algorithm. Providing the ultimate $(2+\epsilon)$ approximation turns out to be considerably more technical. 

Assume we have a graph of size $O(n \log{\log{n}})$ that preserves all the distances up to a $(1+\epsilon,\beta)$-stretch for $\beta = O(\frac{\log{\log{n}}}{\epsilon})^{\log{\log{n}}}$. Since this graph is sparse enough, all vertices can learn it in $O(\log{\log{n}})$ rounds, which already gives a $(1+\epsilon,\beta)$-approximation for APSP. For paths of length at least $t = O(\frac{\beta}{\epsilon})$, this actually gives a $(1+\Theta(\epsilon))$-approximation. % which can become $(1+\epsilon)$-approximation by rescaling of $\epsilon$. 
%Hence, for long paths we already have a good approximation. 
It therefore remains to provide a near-exact approximation for all vertex pairs at distance at most $t = O(\frac{\beta}{\epsilon})$. The key observation is that since $t$ is small, one can get considerably faster algorithms for this task. Concretely,  we reconstruct the tool-kit of \cite{DBLP:conf/podc/Censor-HillelDK19}, and turn them into distance-sensitive primitives, i.e., with a running time that depends on the given distance threshold $t$. 
For example, \cite{DBLP:conf/podc/Censor-HillelDK19} show $\poly(\log{n})$ algorithms for computing the $k$-nearest vertices and for constructing hopsets.\footnote{For a formal definition of a hopset, see Section \ref{sec:preliminaries}.} We show how to implement $t$-\emph{bounded} variants of these tools in just $\poly(\log t)$ time. Intuitively, if we only want to compute the $k$-nearest vertices at distance at most $t$, we need to employ only $\poly(\log t)=\poly(\log\log n)$ steps of matrix multiplications, rather than $\poly(\log n)$.

Using these tools together with the emulator leads to $\poly(\log{\log{n}})$ algorithm for $(1+\epsilon)$-approximation for multi-source shortest paths. Using the recipe of \cite{DBLP:conf/podc/Censor-HillelDK19}, we can easily convert the $(1+\epsilon)$-MSSP algorithm into a $(3+\epsilon)$-APSP algorithm. Providing the ultimate $(2+\epsilon)$-approximation turns out to be more technical, as following directly the approach in \cite{DBLP:conf/podc/Censor-HillelDK19} adds logarithmic factors to the complexity, and to overcome it we add an additional sparsification phase to the algorithm. 

\subsubsection*{Constructing near-additive emulators}

We next give a high-level overview of our emulator construction.
We build $(1+\epsilon,O(\frac{r}{\epsilon})^{r-1})$-emulator with $O(r n^{1+ \frac{1}{2^r}})$ edges.
For the choice $r = \log{\log n}$, we get a near-linear size emulator.

The construction is based on sampling sets $\emptyset = S_{r+1} \subset S_r \subset S_{r-1} \ldots \subset \ldots S_1 \subset S_0=V$
such that for $1 \leq i \leq r$, the set $S_i$ is constructed by adding each vertex of $S_{i-1}$ to $S_i$ with probability $p_i$.
The probabilities $p_i$ are chosen such that the final emulator would have $O(r n^{1+ \frac{1}{2^r}})$ edges.

We add edges to the emulator as follows. 
For $0 \leq i \leq r$, each vertex $v \in S_i \setminus S_{i+1}$, looks at a ball $B(v,\delta_i,G)$ of radius $\delta_i = \Theta(\frac{1}{\epsilon^i})$ around it. If this ball contains a vertex from $S_{i+1}$, it adds an edge to it, and otherwise it adds edges to all vertices in $B(v,\delta_i,G) \cap S_i$.
Intuitively if there are not too many vertices in $S_i$ at distance $\delta_i$ from $v$ it adds edges to them, otherwise there would be a vertex from $S_{i+1}$ there and it adds an edge to it. 

%The description of the algorithm is slightly different from the description in the previous examples in two ways: First, in the examples we decided to add edges according to the number of vertices in $S_i$ in the ball around $v$, which adds an additional logarithmic term to the number of edges. The above variant allows to show that the total number of edges added to the emulator is $O(r n^{1+ \frac{1}{2^r}})$ and not just $\tilde{O}(r n^{1+ \frac{1}{2^r}})$, which is crucial for our applications, as an additional logarithmic term would add a log factor to the complexity of our shortest paths algorithms. Second, in the last phase, in previous examples we added edges from the last set $S_r$ to all other vertices, while this does not affect the parameters of the emulator, this step is difficult to implement efficiently in the \clique model. Intuitively, in all other iterations vertices look at some local ball around them which is a local process, where adding edges from $S_r$ to all other vertices is a much more global task. However, it is enough for the analysis to add edges only from vertices of $S_r$ to other vertices of $S_r$ at distance at most $\delta_r$, which is a local task.

The intuition for the stretch analysis is as follows. We define a notion of $i$-clustered vertices, that in some sense captures the density of neighbourhoods around vertices. A vertex is $i$-clustered if there is a vertex from $S_i$ close-by. For example, all vertices are $0$-clustered, all high-degree vertices are $1$-clustered, all high-degree vertices that also have a vertex from $S_2$ close-by are $2$-clustered and so on. For a formal definition see Section \ref{sec:emulator}.
Now, if we take $u,v \in V$ such that all vertices in $\pi(u,v)$ are at most $i$-clustered we work as follows. We break $\pi(u,v)$ to segments of length $\frac{1}{\epsilon^i}$, in each such segment we show that there is an additive stretch of $\Theta(\frac{1}{\epsilon^{i-1}})$ between the first and last $i$-clustered vertices in the segment $w_1,w_2$. Summing up over all $d(u,v) \epsilon^{i}$ such segments leads to an additive $+\Theta(d(u,v) \epsilon)$ term. To handle the parts of the segments not between $w_1$ and $w_2$ we use an inductive argument. Each one of the iterations of the algorithm adds $+\Theta(d(u,v) \epsilon)$ term, which eventually leads to a stretch of $(1+\Theta(r\epsilon),O(\frac{1}{\epsilon^{r-1}}))$. By rescaling, we get a stretch of $(1+\epsilon, O(\frac{r}{\epsilon})^{r-1})$.
To get a sparse emulator of size $O(n \log{\log{n}})$, we choose $r = \log{\log{n}}$, which gives a stretch of $(1+\epsilon, O(\frac{\log{\log{n}}}{\epsilon})^{\log{\log{n}}}).$
\vspace{-5pt}

\paragraph{Implementation in the \clique model.}
The intuition for the implementation is as follows. During the algorithm vertices inspect their local balls of radius at most $\delta_r = O(\frac{\log{\log{n}}}{\epsilon})^{\log{\log{n}}}$, and add edges to some of the vertices in these balls. Ideally, we would like to exploit the small radius of the balls to get a complexity of $\poly(\log \delta_r) = \poly(\log{\log n})$. 
The key challenge is that it is not clear how vertices can learn their $t$-hop neighborhood in $\log t$ rounds. This sets a major barrier in the case where the $t$-hop ball of a vertex is dense (i.e., containing $\omega(\sqrt{n})$ vertices). 
Our key idea to overcome this barrier is based on separating vertices to heavy and light based on the sparsity level of their $t$-hop ball. For sparse vertices, whose $t$-hop ball has $O(n^{2/3})$ vertices, the algorithm collects the balls in $\poly(\log t)$ rounds using our distance sensitive tool-kit. For the remaining dense vertices, we will be using the fact that a random collection of $O(\sqrt{n})$ vertices $S$ hit their balls. These ideas allow to implement all the iterations of the algorithm, except the last one. In the last iteration vertices from the last set of the emulator $S_r$ should add edges to all vertices of $S_r$ at distance at most $t = \delta_r$, here we exploit the fact that the set $S_r$ is small and use it together with our bounded hopset to get an efficient implementation.

%\textbf{Obstacles of this approach.} \mptodo{Added this paragraph, need to be integrated better:} The key challenge in obtaining an emulator with the desired sparsity level is the following. While the locality parameter of the $(1+\epsilon,\beta)$ emulator is $t=\beta/\epsilon$, it is not clear how nodes can learn their $t$-hop neighborhood in $\log t$ rounds. This sets a major barrier in the case where the $t$-hop ball of a node is dense (i.e., containing $\omega(\sqrt{n})$ nodes). 
%Our key idea to overcome this barrier is based on separating vertices to heavy and light based on the sparsity level of their $t$-hop ball. For sparse nodes, whose $t$-hop ball has $O(\sqrt{n}\log n)$, the algorithm collects the balls in $O(\log t)$ rounds. For the remaining dense nodes, we will be using the fact that a random collection of $O(\sqrt{n})$ nodes $S$ hit their balls. 
%
%t-hop neighborhood in logt
%rounds, so you need to sparsify your problem to get this complexity
%(separating vertices to heavy/light, focusing on finding distances only
%from sqrt(n) sources at the end, this is written in detail in the
%implementation in the clique part).

%\vspace{-10pt}
\subsubsection*{Deterministic algorithms}
%\mptodo{This paragraph appeared before, nut now might not be needed because of the explanation above Lemma \ref{lem:derand}. Nevertheless I like this text so we might want to do some exchange with the text of Lemma \ref{lem:derand}.  "We also provide deterministic variants for all our algorithms. At a high-level, many of the randomized parts in our algorithms are based on hitting set arguments that can be derandomized using a process described in \cite{ParterY18}. However, derandomizing the emulator construction requires a more careful process. Intuitively, using hitting set arguments to derandomize it would add a logarithmic factor to the size of the emulator, and hence also to the complexity of our shortest paths algorithms. To avoid it, we introduce the \emph{soft hitting set} problem, a variant of the hitting set problem that captures more accurately the randomized process required for constructing the sets $S_i$ of the emulator. We show how to derandomize this process, which leads to deterministic variants of all our applications."}
To derandomize our algorithms, we introduce the notion of soft hitting sets, which is roughly defined as follows. The input to the problem is given by two sets of vertices $R, L$, and an integer $\Delta\leq |L|$. Each vertex $u_i$ in $L$ holds a set $S_i \subseteq R$ of size at least $\Delta$. A set $S_i$ is \emph{hit} by a subset of vertices $Z \subseteq R$ if $S_i \cap Z \neq \emptyset$.
A subset of vertices $Z \subseteq R$ is a \emph{soft hitting set} for the $S_i$'s sets if (i) $|Z|=O(n/\Delta)$ and (ii) the total size of all the $S_i$ sets that are not hit by $Z$ is $O(\Delta \cdot |L|)$. The main benefit of this definition over the standard hitting set definition is in property (i). In the hitting set definition, the size of the hitting set is $O(n\log n/\Delta)$ (i.e., larger by an $O(\log n)$ factor). 

The derandomization of the soft hitting set problem is based on the PRGs of Gopalan et al. \cite{gopalan2012better}. These PRGs can fool a family of DNF formula of $n$ variables with a seed of length $O((\log\log n)^3 \log n)$. Parter and Yogev \cite{ParterY18} observed that the covering conditions of the hitting set problem can be stated as a read-once DNF formula, and used these PRGs to compute hitting sets in $O((\log\log n)^3)$ \clique-rounds.
We extend this framework and show that the soft hitting problem can be stated as a \emph{function} that depends on the probability that a certain DNF formula is satisfied, which allows us to compute the soft hitting set in $O((\log\log n)^3)$ rounds. Overall, the derandomization adds $O((\log{\log{n}})^4)$ term to the complexity of our algorithms. %\mptodo{We may want to give a short definition of the soft-hitting set problem here.} 

We note that under the unbounded local computation assumption, our deterministic bounds can in fact match the randomized ones. Although this assumption is considered to be legit in the \clique model (as well as in other standard distributed models), it is clearly less desirable. Thus in our main constructions we prefer to avoid it by paying an extra $\poly(\log\log n)$ term. 
Nevertheless for completeness, we briefly explain the idea of such a optimal derandomization assuming that the vertices can perfrom unbounded local computation. 
The well-known PRGs by Nisan and Wigderson \cite{NisanW94} provide a logarithmic seed for our probabilistic arguments. Since in our setting, we can compute a chunk of $\lfloor \log n \rfloor$ random bits in $O(1)$ rounds, the entire seed can be computed in $O(1)$ rounds. The only caveat of this PRG is its construction time. However, since in our algorithms the PRG is computed locally, the vertices can use this PRG in black-box in our derandomization algorithms.

\paragraph{Comparison to existing nearly-additive emulator constructions.} 
The most related constructions to ours are the centralized emulator constructions of Elkin and Neiman \cite{elkin2018efficient} and Thorup and Zwick \cite{thorup2006spanners}. The first construction also has distributed implementations in the CONGEST model in $O(n^{\rho})$ time where $\rho$ is an arbitrary small constant \cite{elkin2019near, elkin2018efficient}. For the sake of the discussion, we say that an algorithm is \emph{local} if its exploration radius is \emph{sub-polynomial} %\mtodo{I changed linear to sub-polynomial, as any polynomial is problematic.} 
(i.e., every vertex explores a sub-polynomial ball around it in order to define its output), and otherwise it is \emph{global}. Also, we say that an algorithm is \emph{cluster-centric} if in the clustering procedure of the algorithm, all vertices in the cluster make a collective decision, on behalf of the cluster. An algorithm is \emph{vertex-centric}, if each vertex makes its own individual decisions.%\mtodo{maybe we can remove the notion of clusters (our algorithm doesn't really build any clusters), and just say that a vertex does its own decisions, what do you think?}
 
With this rough characterization in hand, the EN algorithm is local and cluster-centric. The TZ algorithm, on the other hand, is global and vertex-centric. Our algorithm appears to be an hybrid of the two: it is local and vertex-centric. Both of these properties are important for its efficient implementation. The vertex-centric approach allows to give a very short and simple algorithm, this led to a very fast implementation in the \clique model. The local approach is crucial to obtaining a sub-logarithmic complexity. Implementing the global approach of TZ in the \clique model seems to require at least poly-logarithmic time.
We remark that while the description of our algorithm is different from the description of EN, they seem to be two different ways of viewing a similar process. This connection is useful for the analysis, and indeed some elements in our stretch analysis are inspired by \cite{elkin2019near, elkin2018efficient}. For a more detailed comparison between the algorithms, see Appendix \ref{sec:comparison}. For a recent survey on near-additive spanners and emulators see \cite{EN-HopsetSurvey20}. 

\section{Preliminaries} \label{sec:preliminaries}

To implement the algorithm, we need the following definitions and tools.

\paragraph{Notation.} Given a graph $G = (V,E)$, we denote by $d_G(u,v)$ the distance between the vertices $u$ and $v$ in $G$. If $G$ is clear from the context, we use the notation $d(u,v)$ for $d_G(u,v).$ The $t$-hop distance between $u$ and $v$ in $G$, denoted by $d^t_G(u,v)$, is the distance of the shortest path between $u$ and $v$ that uses at most $t$ edges. 
We denote by $B(v,\delta,G)$ the ball of radius $\delta$ around $v$ in $G$. 
We use the term w.h.p for an event that happens with probability at least $1-\frac{1}{n^c}$ for some constant $c$.
We use the notation $S' \gets Sample(S, p)$ for a subset $S'$ that is sampled from $S$ by adding each vertex to $S'$ independently with probability $p$.

\paragraph{Model.} 
In the \clique model, we have a communication network of $n$ vertices, communication happens in synchronous rounds, and per round, each two vertices exchange $O(\log n)$ bits.
The input and output are local, in the sense that initially each of the $n$ vertices knows only its incident edges in the graph $G=(V, E)$, and at the end each vertex should know the part of the output adjacent to it. For example, its distances from other vertices.
%\mtodo{moved here.}

\paragraph{Near-additive emulator.} Given an \emph{unweighted} graph $G=(V,E)$, a \emph{weighted} graph $H=(V,E')$ on the same set of vertices is a $(1+\epsilon,\beta)$-emulator for $G$ if for any pair of vertices $u,v \in V$, it holds that $d_G(u,v) \leq d_H(u,v) \leq (1+\epsilon)d_G(u,v) + \beta.$

\paragraph{Hitting sets.} 
Let $S_v \subseteq V$ be a set of size at least $k$. We say that $A$ is a hitting set of the sets $\{S_v\}_{v \in V}$, if $A$ has a vertex from each one of the sets $S_v.$
We can construct hitting sets easily by adding each vertex to $A$ with probability $p = O(\frac{\log{n}}{k})$, which gives the following.
For a proof see Appendix \ref{sec:proof_hitting}.

\begin{restatable}{lemma}{rhitting}\label{rand_hit}
Let $V' \subseteq V$, and let $\{S_v \subseteq V\}_{v \in V'}$ be a set of subsets of size at least $k$.
There exists a \emph{randomized} algorithm in the \clique model that constructs a hitting set of size $O(n\log{n}/k)$ w.h.p, without communication. 
\end{restatable}

A recent \emph{deterministic} construction of hitting sets in the \clique model is given in \cite{ParterY18}, which show the following (see Corollary 17 in \cite{ParterY18}).

\begin{lemma} \label{det_hit}
Let $V' \subseteq V$, and let $\{S_v \subseteq V\}_{v \in V'}$ be a set of subsets of size at least $k$, such that $S_v$ is known to $v$.
There exists a \emph{deterministic} algorithm in the \clique model that constructs a hitting set of size $O(n\log{n}/k)$ in $O((\log{\log{n}})^3)$ rounds. 
\end{lemma}

\paragraph{The \kdnearest problem.} In the \kdnearest problem, we are given integers $k,d$, and the goal of each vertex is to learn the distances to the closest $k$ vertices of distance at most $d$. If there are less than $k$ vertices at distance at most $d$ from $v$, it learns the distances to all vertices of distance at most $d$. We call the $k$ closest vertices of distance at most $d$ from $v$, the \kdnearest vertices to $v$, and denote them by $N_{k,d}(v)$. We focus on solving this problem in \emph{unweighted} graphs.

This problem is an extension of the \knearest problem discussed in \cite{DBLP:conf/podc/Censor-HillelDK19}, where the goal is to find the $k$-nearest vertices in the whole graph, without restriction on the distance. Extending ideas from \cite{DBLP:conf/podc/Censor-HillelDK19}, we show that the \kdnearest problem can be solved in time $poly(\log{d})$ if we only consider distances bounded by $d$ and $k=O(n^{2/3})$. The proof appears in Section \ref{sec:nearest}.

\begin{restatable}{theorem}{nearest}\label{thrm:knearest}
The \kdnearest problem can be solved in unweighted graphs in
\[O \left(\left( \frac{k}{n^{2/3} } + \log d \right) \log  d \right)\]
rounds in \clique.
\end{restatable}

\paragraph{The \sdk problem.} 
In the \sdk problem, we are given a set of \emph{sources} $S \subseteq V$ and an integer $d$, and the task is to compute for each vertex $v$ the set of sources within hop distance at most $d$, as well as the distances to those sources using paths of at most $d$ hops. The input graph to the problem can be \emph{weighted}. The following is shown in \cite{DBLP:conf/podc/Censor-HillelDK19}.\footnote{In \cite{DBLP:conf/podc/Censor-HillelDK19}, the authors consider a more general variant, called $(S,d,k)$-source detection, where it is required to compute only distances to the $k$ closest sources of hop distance at most $d$. In our applications, $k=|S|$.}

\begin{theorem}\label{thrm:source-detection}
The \sdk problem can be solved in weighted graphs in
\[ O \biggl(\biggl( \frac{ m^{1/3} |S|^{2/3}}{n} +1 \biggr) d \biggl) \]
rounds in \clique, where $m$ is the number of edges in the input graph.
\end{theorem}

\paragraph{Bounded Hopsets.}

Note that in Theorem \ref{thrm:source-detection} there is a linear dependence on $d$, hopsets would allow us to solve the same task in time $poly(\log{d})$ at the price of obtaining $(1+\epsilon)$-approximations for the distances. 

Given a graph $G$, a $(\beta,\epsilon)$-hopset $H$ is a weighted graph on the same vertices, such that the $\beta$-hop distances in $G \cup H$ give $(1+\epsilon)$-approximation for the distances in $G$. The importance of hopsets comes from the fact that they allow to focus only on short paths in $G \cup H$ which is crucial for obtaining a fast algorithm. In \cite{DBLP:conf/podc/Censor-HillelDK19} it is shown that hopsets can be constructed in $O(\frac{\log^2{n}}{\epsilon})$ rounds in the \clique, which is too expensive for our purposes. Here, we focus on constructing \emph{bounded} hopsets, that give good approximation only to paths of at most $t$ hops, and show that they can be constructed in time which is just $poly(\log{t})$. 

Bounded hopsets are defined as follows. Given a graph $G$, a $(\beta,\epsilon, t)$-hopset $H$ is a weighted graph on the same vertices, such that the $\beta$-hop distances in $G \cup H$ give $(1+\epsilon)$-approximation for all pairs of vertices in $G$ where $d_G(u,v)=d^t_G(u,v)$. In unweighted graphs, these are all pairs of vertices at distance at most $t$ from each other. For any such pair, we have $d_G(u,v) \leq d_{G \cup H}(u,v)$ and $$d_G(u,v) \leq d_{G \cup H}^{\beta}(u,v) \leq (1+\epsilon)d_G(u,v).$$

Extending ideas from \cite{DBLP:conf/podc/Censor-HillelDK19}, we show the following. The proof appears in Section \ref{sec:hopsets}.

\begin{restatable}{theorem}{hopset}
\label{hopset_thm}
Let $G$ be an unweighted undirected graph and let $0< \epsilon < 1$. 
\begin{enumerate}
\item There is a randomized construction of a $(\beta,\epsilon, t)$-hopset with $O(n^{3/2} \log{n})$ edges and $\beta=O\bigl((\log{t})/\epsilon\bigr)$ that takes $O\bigl((\log^2{t})/\epsilon \bigr)$ rounds w.h.p in the \clique model.
\item There is a deterministic construction of a $(\beta,\epsilon, t)$-hopset with $O(n^{3/2} \log{n})$ edges and $\beta=O\bigl((\log{t})/\epsilon\bigr)$ that takes $O\bigl((\log^2{t})/\epsilon + (\log{\log{n}})^3 \bigr)$ rounds in the \clique model.
\end{enumerate}
\end{restatable}

\section{$(1+\epsilon,\beta)$-Emulators}\label{sec:emulator}

\subsection{Warm-up: $(1+\epsilon,\Theta(\frac{1}{\epsilon}))$-emulator with $\tilde{O}(n^{1+ \frac{1}{4}})$ edges.}

To illustrate the intuition behind our algorithm, we start by describing a simplified variant with $\tilde{O}(n^{1+\frac{1}{4}})$ edges.
The construction samples two random sets $S_1,S_2$ as follows: $S_1$ is a random set of size $O(n^{3/4})$, defined by adding each vertex to $S_1$ with probability $\frac{1}{n^{1/4}}$, and $S_2$ is a random set of size $O(n^{1/4})$ defined by adding each vertex of $S_1$ to $S_2$ with probability $\frac{1}{n^{1/2}}.$
We add to the emulator the following (possibly weighted) edges: 
\begin{enumerate}
\item All edges adjacent to low-degree vertices of degree at most $n^{1/4} \log{n}$. In addition, each vertex of degree at least $n^{1/4} \log{n}$ adds an edge to a neighbour in $S_1$, which exists w.h.p, as $S_1$ is a random set of size $n^{3/4}.$\label{low_degree}
\item Vertices $v \in S_1$ consider their ball $B(v, \delta, G)$ of radius $\delta = \frac{1}{\epsilon} + 2$ around $v$. If this ball has at most $\sqrt{n}\log{n}$ vertices from $S_1$, $v$ adds an edge to each of them. Otherwise, w.h.p, this ball has a representative from $S_2$, and $v$ adds a weighted edge to it.\label{S1_edges}
\item The vertices in $S_2$ add edges to all the vertices.\label{S2_edges}
\end{enumerate}
Each edge $\{u,v\}$ added to the emulator has a weight of $d_G(u,v)$. % the shortest path between $u$ and $v$.
Note that in Line \ref{low_degree} each vertex adds at most $n^{1/4} \log{n}$ edges. In Line \ref{S1_edges}, each vertex of $S_1$ adds at most $\sqrt{n}\log{n}$ edges, as $|S_1|=O(n^{3/4})$ is sums up to $\tilde{O}(n^{1+ \frac{1}{4}})$ edges. Finally, in Line \ref{S2_edges}, we add at most $O(n^{1+\frac{1}{4}})$ edges as $|S_2| = n^{1/4}$. Hence, in total the emulator has $\tilde{O}(n^{1+ \frac{1}{4}})$ edges.\\[-7pt]

\textbf{Stretch analysis.} We next sketch the stretch analysis. Let $u,v \in V$ such that $\pi(u,v)$ is the shortest $u-v$ path. If $\pi(u,v)$ only contains low degree vertices it is contained in the emulator. Otherwise, $\pi(u,v)$ has at least one vertex $w$ with degree at least $n^{1/4} \log{n}$. As $S_1$ is a random set of size $O(n^{3/4})$, then $w$ has a neighbour $s \in S_1$. There are 2 options for $s$, either it adds edges to all vertices from $S_1$ of radius $\delta = \frac{1}{\epsilon} + 2$ around it, or it adds one edge to a vertex in $S_2$ at distance at most $\delta$. In the first case, we say that $w$ is $1$-clustered and in the second case we say that $w$ is $2$-clustered. We break into cases according to this.
The case that $w$ is 2-clustered is actually easy, as in this case there if a vertex from $S_2$ at distance $\Theta(\frac{1}{\epsilon})$ from $\pi(u,v)$, since we added edges between all vertices to vertices in $S_2$, it follows that the emulator has a $u-v$ path with additive stretch $\Theta(\frac{1}{\epsilon})$.
The more interesting case is that all high-degree vertices are $1$-clustered, in this case we can break the path to sub-paths of length $\frac{1}{\epsilon}$ and show that in each of them we get an additive $+4$ stretch. The intuition is that if we take the first and last high degree vertices in each such sub-path, the emulator has a short path between them since their neighbours in $S_1$ have an edge between them. As $\pi(u,v)$ has $\epsilon d(u,v)$ segments of length $\frac{1}{\epsilon}$, this results in a $(1+4\epsilon, 4)$-stretch overall, where the last additive term is for the case that $d(u,v) < \frac{1}{\epsilon}.$ If we sum up over all cases, we get a stretch of $(1+4\epsilon, \Theta(\frac{1}{\epsilon})).$ By rescaling of $\epsilon$ we can get a $(1+\epsilon, \Theta(\frac{1}{\epsilon}))$ stretch.\\[-7pt]

\subsection{Algorithm for $(1+\epsilon,\beta)$-Emulators} \label{sec:alg}

We next describe an algorithm that builds an emulator of size $O(n \log{\log{n}})$ edges and stretch of $(1+\epsilon, O(\frac{\log{\log{n}}}{\epsilon})^{\log{\log{n}}}).$
Later we show a variant of the algorithm that can be efficiently implemented in the \clique.

Let $\emptyset = S_{r+1} \subset S_r \subset S_{r-1} \ldots \subset \ldots S_1 \subset S_0=V$ be random subsamples of vertices, such that 
$$S_i \gets Sample(S_{i-1}, p_i) \mbox{~~for every ~~} i \in \{1, \ldots, r\},$$
where $p_i = n^{-\frac{2^{i-1}}{2^r}}$ for $1 \leq i \leq r-1$, and $p_r = n^{-\frac{1}{2^r}}$.\footnote{The definition of $p_r$ would be useful for the \clique implementation, for other purposes it is possible to define it in the same way as the other $p_i$ values.} 

For $i \in \{0, \ldots, r\}$, we define the values $R_i, \delta_i$. Intuitively, $\delta_i$ is a radius such that vertices in $S_i$ consider the ball of radius $\delta_i$ around them, and add edges to certain vertices there, and $R_i$ would bound the distance between certain vertices to the closest vertex in $S_i$. The values $R_i, \delta_i$ are defined as follows. $R_0 = 0, R_i = \sum_{j=0}^{i-1} \delta_j$, where $\delta_i = \frac{1}{\epsilon^i} + 2R_i.$
For every $i \in \{0,\ldots, r\}$, define $s \in S_i$ to be $i$-dense if 
$$B(s, \delta_i ,G)\cap S_{i+1} \neq  \emptyset,$$
and otherwise it is $i$-sparse. The emulator would include edges from each $i$-sparse vertex $s \in S_i$ to all vertices in $S_i$ of distance at most $\delta_i$, and one edge from each $i$-dense vertex $s \in S_i$ to a vertex in $S_{i+1}$ of distance at most $\delta_i$.

For each vertex $v$ and $i \in \{0,\ldots, r\}$, we define $c_i(v)$ to be a vertex in $S_i$ of distance at most $R_i$ from $v$ if such exists. $c_i(v)$ is defined inductively as follows. Set $c_0(v)=v$. The vertex $c_{i+1}(v)$ exists only if $c_i(v)$ exists and is $i$-dense. In this case, we define $c_{i+1}(v)$ to be the closest vertex to $c_i(v)$ in $B(c_i(v), \delta_i ,G)\cap S_{i+1}$, such vertex exists since $c_i(v)$ is $i$-dense. We say that a vertex is $i$-clustered if $c_i(v)$ exists. In particular, all vertices are 0-clustered, and each $i$-clustered vertex is also $i'$-clustered for any $0 \leq i' \leq i$.

\paragraph{The edges of the emulator.}
For every $v \in S_i \setminus S_{i+1}$, we add to the emulator $H$ the following edges:
\begin{itemize}
\item
If $v$ is $i$-dense, we add one edge to $c_{i+1}(v)$.
\item
If $v$ is $i$-sparse, we add edges to all vertices $u \in B(v,\delta_i,G) \cap S_i.$
\end{itemize}
For each edge $\{u,v\}$ added to the emulator, its weight is the weight of the shortest $u-v$ path in $G$.
Note that each vertex is exactly in one of the sets $S_i \setminus S_{i+1}$, hence it adds edges only once in the process. 
We next show that the distance between $v$ to $c_i(v)$ in the emulator is indeed at most $R_i$.

\begin{claim} \label{claim_Ri}
Every $i$-clustered vertex $v$ is connected to $c_i(v)$ in the emulator by a path of at most $i$ edges and of total weight at most $R_i$.
\end{claim}

\begin{proof}
The proof is by induction. For $i=0$, it holds that $c_0(v)=v,$ hence the claim clearly holds. Assume it holds for $i$, and we prove it for $i+1$. Let $v$ be an $(i+1)$-clustered vertex, then $v$ is also $i$-clustered, and from the induction hypothesis, a path of total weight at most $R_i$ between $v$ and $c_i(v)$ exists in the emulator, and the path has at most $i$ edges. Since $v$ is $(i+1)$-clustered, then $c_i(v)$ is $i$-dense. One case it that $c_i(v) \in S_{i+1}$, in this case $c_{i+1}(v) = c_i(v),$ hence the claim clearly holds. Otherwise, $c_i(v) \in S_i \setminus S_{i+1},$ and since $c_i(v)$ is $i$-dense it adds an edge to a vertex in $S_{i+1}$ of distance at most $\delta_i$. By the definition, this vertex is $c_{i+1}(v)$. Hence, there is a path of at most $i+1$ edges in the emulator between $v$ and $c_{i+1}(v)$ of weight at most $R_i + \delta_i = \sum_{j=0}^{i-1} \delta_j + \delta_i = \sum_{j=0}^{i} \delta_j = R_{i+1}$, as needed.
\end{proof}

\subsection{Size analysis}

We next bound the size of the emulator, simple inductive arguments show the following, for a proof see Appendix \ref{sec:emulator_app_size}.

\begin{restatable}{claim}{sizeSi}\label{size_Si}
For $0 \leq i <r$, it holds that $S_i$ is of size $n^{1-\frac{2^i-1}{2^r}}$ in expectation. 
\end{restatable} 

\begin{restatable}{claim}{probSr} \label{claim_prob_Sr}
Each vertex is in $S_r$ with probability $\frac{1}{\sqrt{n}}$.
\end{restatable}

\begin{restatable}{claim}{sizeSr} \label{size_Sr}
$S_r$ is of size $O(\sqrt{n})$ w.h.p.
\end{restatable}

We next show the following.

\begin{claim} \label{claim_edges}
For any $0 \leq i<r$, each vertex $s \in S_i \setminus S_{i+1}$ adds at most $O(n^{\frac{2^i}{2^r}})$ edges to the emulator in expectation.
\end{claim}

\begin{proof}
We will show that a vertex $s \in S_i \setminus S_{i+1}$ adds $O(\frac{1}{p_{i+1}})$ edges to the emulator in expectation. As for $0 \leq i < r-1$, we have $\frac{1}{p_{i+1}} = n^{\frac{2^i}{2^r}}$, and for $i=r-1$, we have $\frac{1}{p_r} = n^{\frac{1}{2^r}} < n^{\frac{2^{r-1}}{2^r}}$, the claim follows.

Let $x=\frac{1}{p_{i+1}}$. By definition, $S_{i+1}$ is sampled from $S_i$ with probability $p_{i+1} = \frac{1}{x}$. Hence, intuitively, vertices $s \in S_i$ with $|B(s,\delta_i,G) \cap S_i| \geq x$ would have vertex from $S_{i+1}$ at distance $\delta_i$ and would only add one edge to the emulator, other vertices in $S_i$ would only add at most $x$ edges to the emulator. We next formalize this intuition. Let $s \in S_i \setminus S_{i+1}$. If $|B(s,\delta_i,G) \cap S_i| \leq x$, then $s$ adds edges to at most $x$ vertices as needed. Assume now that $|B(s,\delta_i,G) \cap S_i| = \alpha x$ for $\alpha > 1$. The probability that no vertex from $B(s,\delta_i,G)$ is sampled for $S_{i+1}$ is $(1-\frac{1}{x})^{\alpha x}$, in this case it adds to the emulator $\alpha x$ edges in iteration $i$, and otherwise it adds one edge. Hence, the expected number of edges added by $s$ is bounded by $1+\alpha x (1-\frac{1}{x})^{\alpha x}$. Since $(1-\frac{1}{x})^{\alpha x}$ goes to $\frac{1}{e^\alpha}$ and $\frac{\alpha}{e^{\alpha}} < 1$, the expected number of edges is $O(x)$.
\end{proof}

\begin{lemma} \label{lemma_size} 
The emulator $H$ has $O(r \cdot n^{1+\frac{1}{2^r}})$ edges in expectation.
\end{lemma}

\begin{proof}
By Claim \ref{size_Si}, for $0 \leq i <r$, $S_i$ is of size $n^{1-\frac{2^i-1}{2^r}}$ in expectation, and by Claim \ref{claim_edges}, each vertex $s \in S_i \setminus S_{i+1}$ for $0 \leq i<r$ adds $O(n^{\frac{2^i}{2^r}})$ edges to the emulator in expectation. Hence, for $0 \leq i<r$, the expected number of edges added by all the vertices of $S_i \setminus S_{i+1}$ is bounded by 
$$O( n^{1-\frac{2^i-1}{2^r}} \cdot n^{\frac{2^i}{2^r}}) = O(n^{1+\frac{1-2^i+2^i}{2^r}})=O(n^{1+\frac{1}{2^r}}).$$
For $i=r$, since $S_r$ is of size $O(\sqrt{n})$ w.h.p by Claim \ref{size_Sr}, even if all vertices of $S_r$ add edges to all vertices of $S_r$ this adds at most $O(n)$ edges to the emulator w.h.p. Summing over all $i$, gives a total of $O(r \cdot n^{1+\frac{1}{2^r}})$ edges in expectation.
\end{proof}

\subsection{Stretch analysis}

The stretch analysis appears in Lemma \ref{lemma_stretch}. To prove it, we need the following technical claims, for proofs see Appendix \ref{sec:emulator_app_stretch}.
We start by bounding the value of $R_i$.

\begin{restatable}{claim}{stretchRi}
$R_i = \sum_{j=0}^{i-1} \frac{1}{\epsilon^j} \cdot 3^{i-1-j}.$
\end{restatable}

\begin{restatable}{claim}{Ribound} \label{claim_Ri_bound}
Let $0< \epsilon < \frac{1}{6}$, then $R_i \leq \frac{2}{\epsilon^{i-1}}.$
\end{restatable}

We next define a value $\beta_i$, that would give a bound on the additive stretch of paths of $i$-clustered vertices in Lemma \ref{lemma_stretch}. It is defined as follows.
Define $\beta_0 = 0, \beta_i = 4 \sum_{j=1}^{i} 2^{i-j} R_j.$ 

\begin{restatable}{claim}{Stretchtwo} \label{claim_stretch_2}
$4R_i + 2 \beta_{i-1} = \beta_i.$
\end{restatable}

\begin{restatable}{claim}{Stretchone} \label{claim_stretch_1}
Let $0< \epsilon < \frac{1}{10}$, then $\beta_i \leq \frac{10}{\epsilon^{i-1}}.$
\end{restatable}

Let $\pi(u,v)$ be any shortest $u-v$ path. We denote by $d_H(u,v)$ the distance between $u$ and $v$ in the emulator, and by $d(u,v)$ the distance between them in $G$. 

\begin{lemma} \label{lemma_stretch}
Assume that all the vertices in $\pi(u,v)$ are at most $i$-clustered, except maybe one that is $j$-clustered for $j>i$, then $d_H(u,v) \leq (1+20 \epsilon i)d(u,v)+ \beta_i.$ 
\end{lemma}

\begin{proof}
The proof is by induction. If $i=0$, all the vertices in $\pi(u,v)$ are at most $0$-clustered, except maybe one. As all vertices that are at most $0$-clustered are $0$-sparse and add all their adjacent edges to the emulator, it holds that all the edges of $\pi(u,v)$ are taken to the emulator. Hence, $d_H(u,v)=d(u,v)$. Assume that the claim holds for $i-1$, and we show that it holds for $i$.

\emph{Case 1: $\pi(u,v)$ has length at most $\frac{1}{\epsilon^i}$.} If there is only one vertex in $\pi(u,v)$ that is $j$-clustered for $j \geq i$, then the claim follows from the induction assumption. Hence, we focus on the case that there are at least two $i$-clustered vertices in $\pi(u,v)$. Denote the first and last such vertices by $u'$ and $v'$ respectively. From Claim \ref{claim_Ri}, for every $i$-clustered vertex $w$, the distance in the emulator between $w$ and $c_i(w) \in S_i$ is at most $R_i$.
Hence, 
$d_H(u',c_i(u')) \leq R_i$ and $d_H(v',c_i(v')) \leq R_i.$
%there are paths of length at most $R_i$ between $u'$ and $c_i(u')$, and between $v'$ and $c_i(v').$ 
Since $d(u',v') \leq \frac{1}{\epsilon^i}$, it follows that $$d(c_i(u'),c_i(v')) \leq d(c_i(u'),u')+d(u',v')+d(v',c_i(v')) \leq d(u',v')+2R_i \leq \frac{1}{\epsilon^i} +2R_i.$$ Also, from the statement of the claim at least one of $u',v'$ is at most $i$-clustered. Assume w.l.o.g that this is $u'$. This means that $c_i(u')$ adds edges to all vertices of $S_i$ at distance at most $\delta_i = \frac{1}{\epsilon^i}+2R_i$, and in particular there is an edge between $c_i(u')$ and $c_i(v')$ in the emulator of weight at most $d(u',v')+2R_i.$ It follows that 
$$d_H(u',v') \leq d_H(u',c_i(u'))+d_H(c_i(u'),c_i(v'))+d_H(c_i(v'),v') \leq R_i + d(u',v')+2R_i + R_i = d(u',v')+4R_i.$$
Now, the subpaths of $\pi(u,v)$ between $u$ and $u'$ and between $v'$ and $v$ are also shortest paths $\pi(u,u'),\pi(v',v)$, which have
only one vertex that is $j$-clustered for $j \geq i$. Hence from the induction hypothesis,
$$d_H(u,u') \leq (1+20\epsilon(i-1))d(u,u')+ \beta_{i-1},$$
$$d_H(v',v) \leq (1+20\epsilon(i-1))d(v',v)+ \beta_{i-1}.$$
Combining it all, we get that
$$d_H(u,v) \leq d_H(u,u') + d_H(u',v') + d_H(v',v) \leq (1+20\epsilon(i-1))d(u,v) +4R_i + 2 \beta_{i-1}.$$
From Claim \ref{claim_stretch_2}, we get $$d_H(u,v) \leq (1+20\epsilon(i-1))d(u,v) + \beta_i,$$ which concludes the proof for this case.

\emph{Case 2: $\pi(u,v)$ has length greater than $\frac{1}{\epsilon^i}$.} Here we divide $\pi(u,v)$ to subpaths of length exactly $\lfloor \frac{1}{\epsilon^i} \rfloor$, except maybe one of length at most $\lfloor \frac{1}{\epsilon^i} \rfloor$. In each of these subpaths we can follow the analysis from Case 1 and get a stretch of $(1+20\epsilon(i-1),\beta_i)$ for the subpath. From Claim \ref{claim_stretch_1}, $\beta_i \leq \frac{10}{\epsilon^{i-1}}$. Hence, for each subpath of length $\lfloor \frac{1}{\epsilon^i} \rfloor$ we add at most additive stretch of $\frac{10}{\epsilon^{i-1}}$. As $\epsilon \leq \frac{1}{2}$, we have that $\lfloor \frac{1}{\epsilon^i} \rfloor \geq \frac{1}{\epsilon^i} -1 \geq \frac{1}{2 \epsilon^i}.$ Hence, there are at most $d(u,v) \cdot 2 \epsilon^i$ subpaths of length exactly $\lfloor \frac{1}{\epsilon^i} \rfloor.$
As each one of them adds an additive stretch of at most $\frac{10}{\epsilon^{i-1}}$, this adds $\frac{10}{\epsilon^{i-1}} \cdot d(u,v)\cdot 2\epsilon^i = 20 \epsilon d(u,v)$ to the total stretch. The last subpath of length at most $\lfloor \frac{1}{\epsilon^i} \rfloor$ adds $\beta_i$ to the additive stretch. In total we get
$$d_H(u,v) \leq (1+20\epsilon(i-1))d(u,v) + 20 \epsilon d(u,v) + \beta_i \leq (1+20\epsilon i)d(u,v) + \beta_i,$$ which completes the proof.
\end{proof}

\paragraph{Conclusion.}

The analysis in previous sections gives the following.

\begin{theorem} \label{thm_conclusion_emulator}
Let $G$ be an unweighted undirected graph, and let $0<\epsilon<1$, there is a randomized algorithm that builds an emulator $H$ with $O(r \cdot n^{1+\frac{1}{2^r}})$ edges in expectation, and stretch of $(1+\epsilon, O(\frac{r}{\epsilon})^{r-1}).$ For the choice $r=\log{\log{n}}$, we have
$O(n \log{\log{n}})$ edges in expectation and stretch of $(1+\epsilon, O(\frac{\log{\log{n}}}{\epsilon})^{\log{\log{n}}}).$
\end{theorem}

\begin{proof}
From Lemma \ref{lemma_size}, the emulator has $O(r \cdot n^{1+\frac{1}{2^r}})$ edges in expectation, and from Lemma \ref{lemma_stretch}, for all pairs of vertices the stretch is at most $(1+20 \epsilon r)d(u,v)+ \beta_r$, as all vertices are at most $r$-clustered. 
Let $\epsilon' = 20 \epsilon r$. Since $\beta_r \leq \frac{10}{\epsilon^{r-1}}$ by Claim \ref{claim_stretch_1}, we have $\beta_r \leq \frac{10}{(\frac{\epsilon'}{20 r})^{r-1}}=O(\frac{r}{\epsilon'})^{r-1}$, which gives a stretch of $(1+\epsilon', O(\frac{r}{\epsilon'})^{r-1}).$ 
For the specific choice of $r=\log{\log{n}}$, we have that the number of edges is $O(\log{\log{n}} \cdot n^{1+\frac{1}{\log{n}}})=O(n \log{\log{n}})$, and the stretch is $(1+\epsilon', O(\frac{\log{\log{n}}}{\epsilon'})^{\log{\log{n}}}).$
\end{proof}

\subsection{Constructing emulators in the \clique}

We next explain how to implement the algorithm from Section \ref{sec:alg} efficiently in the \clique model. 
This algorithm is composed of two parts, first we sample the sets $S_i$, and then each vertex $v \in S_i$ looks at a certain $\delta_i$-neighbourhood around it and adds edges to certain vertices there. The sampling is a completely local task that can be simulated locally by each vertex. For the second part, we focus first on vertices that are in $S_i \setminus S_{i+1}$ for $i<r$. Such vertices are either $i$-sparse, in which case they add edges to all vertices in $S_i$ in radius $\delta_i$, or are $i$-dense in which case they add one edge to the closest vertex from $S_{i+1}$. To implement it we would like to be able to learn the $\delta_i$-neighbourhood of vertices. 
However, this is only efficient if this neighbourhood is small enough. To overcome it, we break into cases according to the size of the neighbourhood $B(v,\delta_i,G)$. If it does not contain a lot of vertices, we say that the vertex $v$ is \emph{light}, and otherwise we say that $v$ is \emph{heavy}.
In the case that $v$ is light it is actually possible to compute the entire neighbourhood $B(v,\delta_i,G)$ using the \kdnearest algorithm. Otherwise, $v$ is heavy. Here although we cannot learn the entire neighbourhood, we can exploit the fact it is large. Intuitively, as the goal in the original algorithm was to deal with different parts of the graph according to their density, if we realize that a certain neighbourhood around a vertex is already quite large, we can bring this vertex immediately to the final stage of the algorithm. This is done as follows. Note that $S_r$ is a random set of $O(\sqrt{n})$ vertices. Since $v$ is heavy, $B(v,\delta_i,G)$ would contain a vertex from $S_r$ which means that $v$ is $i$-dense. In this case, $v$ only adds one edge to the closest vertex from $S_{i+1}$, and we show that it can compute this vertex using the \kdnearest algorithm. To summarize, all vertices not in $S_r$ can actually add all their adjacent edges to the emulator using the \kdnearest algorithm. To deal with vertices in $S_r$ we use a different approach. Note that since $S_{r+1} = \emptyset$, all the vertices in $S_r$ are $r$-sparse and should add edges to all vertices in $S_r$ in their $\delta_r$-neighbourhood. However, this neighbourhood may be large, and there is no clear way to learn the whole neighbourhood. Here we exploit the fact that $S_r$ is of size $O(\sqrt{n})$, so we only need to compute distances to at most $O(\sqrt{n})$ vertices. We show that using the source detection algorithm and the bounded hopset we can compute approximations to all these distances efficiently.

We next describe the algorithm in detail, and explain the changes needed in the analysis since we compute only approximations to distances in the final stage.

\subsubsection*{Sampling the sets $S_i$} 

We sample the sets $S_i$ exactly as described in Section \ref{sec:alg}, this is a local process computed by each vertex. At the end of the process, each vertex sends one message to all other vertices with the index $i$ such that $v \in S_i \setminus S_{i+1}.$

\subsubsection*{Adding edges to the emulator}

To describe the algorithm, we need the following definitions. For a vertex $v$, let $i_v$ be the maximum index $i$ such that $v \in S_i$, and let $B_v=B(v,\delta_{i_v},G).$ We say that $v$ is \emph{heavy} if $|B_v| > n^{2/3}$ and otherwise it is \emph{light}. 
Let $N_{k,\delta_{i_v}}(v)$ be a set of $k=n^{2/3}$ closest vertices to $v$ of distance at most $\delta_{i_v}$. Note that if $v$ is heavy, there are at least $n^{2/3}$ vertices in the $\delta_{i_v}$-neighbourhood of $v$.
Since $S_r$ contains each vertex with probability $\frac{1}{\sqrt{n}}$, we have the following.

\begin{claim} \label{claim_heavy} 
For all heavy vertices $v$, there is a vertex in $N_{k,\delta_{i_v}}(v) \cap S_r$ w.h.p for $k=n^{2/3}$.
\end{claim}

\begin{proof}
If $v$ is heavy, then $|B_v| > n^{2/3}$, this means that there are more than $n^{2/3}$ vertices at distance at most $\delta_{i_v}$ from $v$, and hence $|N_{k,\delta_{i_v}}(v)| = n^{2/3}$. By Claim \ref{claim_prob_Sr}, each vertex is in $S_r$ with probability $\frac{1}{\sqrt{n}}$. Hence, the probability that $N_{k,\delta_{i_v}}(v) \cap S_r = \emptyset$ is $(1-\frac{1}{\sqrt{n}})^{n^{\frac{2}{3}}} \leq (1-\frac{1}{\sqrt{n}})^{\sqrt{n} c \ln{n}} \leq e^{-c \ln{n}} \leq \frac{1}{n^{c}}$ where $c$ can be any constant. Using union bound, we get that w.h.p for all heavy vertices $v$, there is a vertex in $N_{k,\delta_{i_v}}(v) \cap S_r$.
\end{proof}

We will show that for each vertex with $i_v < r$, we can add all the adjacent edges to the emulator by just learning the \kdnearest vertices, and for vertices in $S_r$ we will use the bounded hopsets to get approximations to the distances. 

\begin{claim} \label{claim_edges_notSr}
All the edges of the emulator with at least one endpoint in $V \setminus S_r$ can be added to the emulator with correct distances in $O(\log^2{\delta_r})$ rounds w.h.p. 
\end{claim}

\begin{proof}
To implement the algorithm, we compute the \kdnearest for $k=n^{2/3}, d=\delta_r$, which takes $O(\log^2{\delta_r})$ time by Theorem \ref{thrm:knearest}. Let $N_{k,d}(v)$ be the set of $k$ closest vertices of distance at most $d$ to $v$ computed by the algorithm. Let $v \not \in S_r$, and let $i_v$ be the maximum index such that $v \in S_i$. If $v$ is light $|B_v|=|B(v,\delta_{i_v},G)| \leq n^{2/3}$. Since $\delta_{i_v} \leq \delta_r = d$, it follows that $B_v \subseteq N_{k,d}(v)$. Hence, $v$ already knows the distances to all vertices in $B_v$, and since it also knows which of them belong to each set $S_i$, it can add all the relevant edges to vertices in $B_v \cap S_i$, as follows. Since $v$ computed the set $B_v$, it knows whether $B_v \cap S_{i+1} \neq \emptyset$, in this case it is $i$-dense, and adds an edge to the closest vertex in $B_v \cap S_{i+1}$. Otherwise, it is $i$-sparse, and adds edges to all vertices in $B_v \cap S_i.$

We next consider the case that $v$ is heavy. Then, by Claim \ref{claim_heavy}, there is a vertex $u \in N_{k,\delta_{i_v}}(v) \cap S_r$ w.h.p.
Now $v$ knows the distances to all vertices in $N_{k,d}$ and in particular to $u$. Also, $u \in S_r \subseteq S_{i+1}$, which means that $v$ is $i$-dense. Hence, $v$ only needs to add one edge to $c_{i+1}(v)$ which is the closest vertex from $S_{i+1}$, since $v$ knows the distance to $u$ and also to all other vertices strictly closer than $u$, it can add the relevant edge as needed. 
\end{proof}

\begin{claim} \label{claim_edges_Sr}
All the edges in the emulator with two endpoints in $S_r$ can be added to the emulator with $(1+\epsilon')$-approximate distances in $O(\frac{\log^2{\delta_r}}{\epsilon'})$ rounds w.h.p.
\end{claim}

\begin{proof}
Vertices in $S_r$ should add edges to all vertices in $S_r$ of distance at most $\delta_r$. For this, we start by running the bounded hopset algorithm with $t = \delta_r$, which takes $O(\frac{\log^2{\delta_r}}{\epsilon'})$ time and constructs a $(\beta,\epsilon',t)$-hopset $H'$ with $\beta = O(\frac{\log t}{\epsilon'})$ by the randomized construction in Theorem \ref{hopset_thm}. By definition, for all pairs of vertices $u,v$ of distance at most $t$ there is a $\beta$-hop path in $G \cup H'$ of length at most $(1+\epsilon')d_G(u,v)$. 
To learn about those paths we run the \sdk algorithm on the graph $G \cup H'$ with $S = S_r, d=\beta = O(\frac{\log{t}}{\epsilon'})$ which takes  $ O \biggl(\biggl( \frac{ n^{2/3} n^{1/3}}{n} +1 \biggr) \frac{\log{\delta_r}}{\epsilon'} \biggl) = O(\frac{\log{\delta_r}}{\epsilon'})$ time w.h.p by Theorem \ref{thrm:source-detection}, as the size of $S_r$ is $O(\sqrt{n})$ w.h.p by Claim \ref{size_Sr}. This gives $(1+\epsilon')$-approximations for the distances to all sources in $S_r$ of distance at most $\delta_r$ as needed, and hence all vertices in $S_r$ can add all the relevant edges to the emulator with approximate distances. The overall time complexity is $O(\frac{\log^2{\delta_r}}{\epsilon'})$ rounds from constructing the hopsets.
\end{proof}

To conclude, we have the following.

\begin{lemma} \label{lemma_time}
The time complexity of the algorithm is $O(\frac{\log^2{\delta_r}}{\epsilon'})$ rounds w.h.p.
\end{lemma}

\begin{proof}
Sampling the sets $S_i$ and informing all vertices about it takes one round, adding edges to the emulator takes $O(\frac{\log^2{\delta_r}}{\epsilon'})$ rounds w.h.p by Claims \ref{claim_edges_notSr} and \ref{claim_edges_Sr}.
\end{proof}

\subsubsection*{Analysis and conclusion}

Since we compute only approximations to distances of edges with both endpoints in $S_r$, this changes slightly the stretch analysis in the final stage, see Appendix \ref{sec:emulator_app_clique} for full details. To conclude, we get a $(1+\epsilon,\beta)$ emulator with $O(r \cdot n^{1+\frac{1}{2^r}})$ edges in $O(\frac{\log^2{\beta}}{\epsilon})$ rounds, where $\beta = O(\frac{r}{\epsilon})^{r-1}$. Later we focus mostly on the case that $r=\log\log{n}$. In this case, $\log{\beta} = O(\log{\log{n}} \cdot \log{\frac{\log{\log{n}}}{\epsilon}})$. If $\epsilon$ is constant we get that a complexity of $O(\frac{\log^2{\beta}}{\epsilon})$ is roughly $O((\log{\log{n}})^2).$

\begin{restatable}{theorem}{ConclusionExp}
Let $G$ be an unweighted undirected graph, let $0<\epsilon<1$ and let $r \geq 2$ be an integer, there is a randomized algorithm that builds an emulator $H$ with $O(r \cdot n^{1+\frac{1}{2^r}})$ edges in expectation, and stretch of $(1+\epsilon, \beta),$ in $O(\frac{\log^2{\beta}}{\epsilon})$ rounds w.h.p, where $\beta = O(\frac{r}{\epsilon})^{r-1}$. For the choice $r = \log{\log{n}}$, we have $O(n \log{\log{n}})$ edges in expectation, and $\beta = O(\frac{\log{\log{n}}}{\epsilon})^{\log{\log{n}}}.$
\end{restatable}

\subsubsection*{A variant that works w.h.p}

In the algorithm described above, the number of edges is $O(n \log{\log{n}})$ in expectation. For our applications, it would be useful to have this number of edges w.h.p, we next describe a variant that obtains this.

By the proof of Lemma \ref{lemma_size}, all vertices not in $S_r$ add $O(r n^{1+\frac{1}{2^r}})$ edges to the emulator in expectation. In addition, by Claim \ref{size_Sr}, the size of $S_r$ is $O(\sqrt{n})$ w.h.p, which means that $|S_r|^2 = O(n)$ w.h.p, hence the number of edges added between vertices in $S_r$ is at most $O(n)$ w.h.p. Also, from Claim \ref{claim_heavy}, for all heavy vertices $v$, there is a vertex in $N_{k,\delta_{i_v}}(v) \cap S_r$ w.h.p. We would like to find a run where all the above events hold. From Markov's inequality, we have that with constant probability the number of edges added by all vertices not in $S_r$ is $O(r n^{1+\frac{1}{2^r}})$, since the other events hold w.h.p, we have that with constant probability all the above events hold. Hence, if we run the algorithm for $O(\log{n})$ times in parallel, w.h.p we have a run where all events hold. In such a run the total number of edges in the emulator is $O(r n^{1+\frac{1}{2^r}})$. We next explain how to implement the parallel runs efficiently. 
For this, we show how to identify a run where all 3 events hold. 

\begin{claim}
In $O(\log^2{\delta_r} + \log{\log{\log{n}}})$ rounds, we can find w.h.p a run of the algorithm where the number of edges added by vertices not in $S_r$ is $O(r n^{1+\frac{1}{2^r}})$, the size of $S_r$ is $O(\sqrt{n})$ and for all heavy vertices $v$, there is a vertex in $N_{k,\delta_{i_v}}(v) \cap S_r$. 
\end{claim}

\begin{proof}
As discussed above, if we run the algorithm $O(\log{n})$ times, then w.h.p we will have a run that satisfies the above. We next explain how to find it. The first part of the algorithm is to sample the sets $S_i$, this is done by all vertices locally, and we now let all vertices simulate this sampling for $O(\log{n})$ different independent runs of the algorithm. At the end of the sampling, each vertex $v$ should inform all other vertices, the index $i$ where $v \in S_i \setminus S_{i+1}$, since $i \leq \log\log{n}$, this number can be represented in $O(\log{\log{\log{n}}})$-bits, and hence to send all the $O(\log{n})$ values for the different runs, we only need $O(\log{\log{\log{n}}})$ rounds, as the size of messages in each round is $O(\log{n})$.

Next, we run the \kdnearest algorithm with $k=n^{2/3},d=\delta_r.$ As explained in the proof of Claim \ref{claim_edges_notSr}, this allows all vertices not in $S_r$ to learn about all the edges they add to the emulator. Note that we only need to run the \kdnearest algorithm once, to simulate all the runs. While the identity of the edges added to the emulator may change between runs, the only relevant edges are to vertices in the \kdnearest, and the choice of the edges depends on the sets $S_i$ that vertices belong to, and all vertices already know which vertices are in which set for each one of the runs. In addition, since all vertices know which vertices are in $S_r$, they can check locally if $|S_r| = O(\sqrt{n})$.  For the last event, since vertices computed the \kdnearest, each heavy vertex $v$ knows if there is a vertex in in $N_{k,\delta_{i_v}}(v) \cap S_r$.

Now, to decide about the optimal run we choose $O(\log{n})$ vertices, each one would be assigned to evaluate one of the runs. For a particular run $j$, all vertices send to the vertex $v_j$ responsible for run $j$, the number of edges they added to the emulator in run $j$, and if they are heavy, they also inform whether $N_{k,\delta_{i_v}}(v) \cap S_r \neq \emptyset$ in run $j$. From the above information, $v_j$ can see if all 3 events hold for run $j$. At the end, all $O(\log{n})$ vertices assigned for the runs let all vertices know the total number of edges added to the emulator by vertices not in $S_r$, and whether the 2 other events hold. Then, all vertices choose an optimal run to be the run with minimum number of edges added, from the ones that satisfy all the events. 

As computing the \kdnearest takes $O(\log^2{\delta_r})$ rounds, the total time complexity is $O(\log^2{\delta_r} + \log{\log{\log{n}}})$, and since we simulate $O(\log{n})$ runs, w.h.p we find a run where all 3 events hold.
\end{proof}

After we identify a run where all 3 events hold, all vertices know in which sets $S_i$ all vertices are in this run, and from now on they can just run the whole algorithm (actually, since we already computed the \kdnearest, all vertices not in $S_r$ already know which edges they add to the emulator, so we only need to run the final stage of the algorithm to decide which edges vertices in $S_r$ add). This gives the following.

\begin{theorem} \label{thm_emulator}
Let $G$ be an unweighted undirected graph, let $0<\epsilon<1$ and let $r \geq 2$ be an integer, there is a randomized algorithm that builds an emulator $H$ with $O(r \cdot n^{1+\frac{1}{2^r}})$ edges w.h.p, and stretch of $(1+\epsilon, \beta),$ in $O(\frac{\log^2{\beta}}{\epsilon})$ rounds w.h.p, where $\beta = O(\frac{r}{\epsilon})^{r-1}$. For the choice $r = \log{\log{n}}$, we have $O(n \log{\log{n}})$ edges w.h.p, and $\beta = O(\frac{\log{\log{n}}}{\epsilon})^{\log{\log{n}}}.$
\end{theorem}

\section{Applications} \label{sec:applications}

\subsection{$(1+\epsilon, \beta)$-approximation of APSP}

Using the emulator, we directly get a near-additive approximation for APSP, by building a sparse emulator and letting all vertices learn it. This gives the following.  

\begin{theorem} \label{thm_near_additive_apsp}
Let $0<\epsilon<1$, there is a randomized $(1+\epsilon, \beta)$-approximation algorithm for unweighted undirected APSP in the \clique model that takes $O(\frac{\log^2{\beta}}{\epsilon})$ rounds w.h.p, where $\beta = O(\frac{\log{\log{n}}}{\epsilon})^{\log{\log{n}}}.$
\end{theorem}

\begin{proof}
First, we use Theorem \ref{thm_emulator} to build a $(1+\epsilon,\beta)$-emulator of size $O(n \log{\log{n}})$ in $O(\frac{\log^2{\beta}}{\epsilon})$ rounds w.h.p, for $\beta = O(\frac{\log{\log{n}}}{\epsilon})^{\log{\log{n}}}$. By definition, the emulator contains a path with $(1+\epsilon,\beta)$ stretch for all pairs of vertices.
Hence, to approximate APSP, we let all vertices learn the emulator. This is done as follows. As the size of the emulator is $O(n \log{\log{n}})$ w.h.p, and each vertex knows the edges adjacent to it in the emulator, we can use Lenzen's routing \cite{lenzen2013optimal} to let one vertex $v$ learn the whole emulator in $O(\log{\log{n}})$ rounds w.h.p. Then, in additional $O(\log{\log{n}})$ rounds all vertices can learn it, as $v$ can divide the edges of the emulator into $n$ parts of size $O(\log{\log{n}})$, send each one of them to one vertex, and then each vertex can send this information to all other vertices in $O(\log{\log{n}})$ time, allowing all vertices learn the whole emulator.
\end{proof}

\subsection{$(1+\epsilon)$-approximation of multi-source shortest paths}
We next show how to get $(1+\epsilon)$-approximation for SSSP or multi-source shortest paths as long as the number of sources is $O(\sqrt{n})$ in just $poly(\log{\log{n}})$ rounds in unweighted undirected graphs. The idea is simple, for vertices that are far away, a $(1+\epsilon,\beta)$-approximation already gives a $(1+\epsilon)$-approximation. So we just need to take care of close by vertices of distance around $O(\frac{\beta}{\epsilon})$ from each other, for this we can use the bounded hopset. We show the following.

\begin{theorem} \label{thm_maap}
Let $0<\epsilon<1$ and let $G$ be an unweighted undirected graph, there is a randomized $(1+\epsilon)$-approximation algorithm for multi-source shortest paths in the \clique model from a set of sources $S$ of size $O(\sqrt{n})$ that takes $O(\frac{\log^2{\beta}}{\epsilon})$ rounds w.h.p, where $\beta = O(\frac{\log{\log{n}}}{\epsilon})^{\log{\log{n}}}.$
\end{theorem}

\begin{proof}
We use Theorem \ref{thm_emulator} to build $(1+\frac{\epsilon}{2}, \beta)$-emulator for $\beta=O(\frac{\log{\log{n}}}{\epsilon})^{\log{\log{n}}}$ in $O(\frac{\log^2{\beta}}{\epsilon})$ rounds w.h.p. Since the size of the emulator is $O(n \log{\log{n}})$ w.h.p, all vertices can learn it in $O(\log{\log{n}})$ rounds, as explained in the proof of Theorem \ref{thm_near_additive_apsp}. 

Let $t = \frac{2\beta}{\epsilon}$, we build a $(h,\epsilon,t)$-hopset $H'$, for $h = O(\frac{\log{t}}{\epsilon})$ in $O(\frac{\log^2{t}}{\epsilon})$ rounds using the randomized algorithm in Theorem \ref{hopset_thm}. This gives us $h$-hop paths with $(1+\epsilon)$-approximate distances to all pairs of vertices of distance at most $t$. Then, we can run the \sdk algorithm to compute for each vertex the distances from $S$ using paths of hop distance at most $h$ in the graph $G \cup H'.$ This takes $ O \biggl(\biggl( \frac{ n^{2/3} |S|^{2/3}}{n} +1 \biggr) h \biggl)=O\biggl(\biggl(\frac{|S|^{2/3}}{n^{1/3}} +1 \biggr)  h \biggr)$ rounds. In particular the complexity is $O(h)=O(\frac{\log{t}}{\epsilon})$ if $|S|=O(\sqrt{n}).$

For each pair of vertices $\{v,u\}$ where $u \in S$, we evaluate the distance with the minimum distance found throughout the computation, which can be either the distance between them in the emulator or in the \sdk algorithm. If $d(u,v) \leq t$, then during the \sdk algorithm $u$ and $v$ learn $(1+\epsilon)$-approximate distance, as by definition, the graph $G \cup H'$ has a $h$-hop path between $u$ and $v$ with $(1+\epsilon)$-approximate distance in this case. Otherwise, $d(u,v) \geq t = \frac{2\beta}{\epsilon}$, in this case, the $(1+\frac{\epsilon}{2},\beta)$-emulator has a path between $u$ and $v$ of distance at most $(1+\frac{\epsilon}{2})d(u,v)+\beta \leq (1+\frac{\epsilon}{2})d(u,v) + \frac{\epsilon}{2} d(u,v) \leq (1+\epsilon)d(u,v).$
Hence, we are guaranteed to find a $(1+\epsilon)$-approximation for the distance $d(u,v).$

The time complexity is $O(\frac{\log^2{\beta}}{\epsilon} + \frac{\log^2{t}}{\epsilon})$ for $|S|=O(\sqrt{n})$ and $t=\frac{2\beta}{\epsilon}$. Since $\beta = O(\frac{\log{\log{n}}}{\epsilon})^{\log{\log{n}}}$, $t$ can be bounded by $O(\frac{\log{\log{n}}}{\epsilon})^{\log{\log{n}}+1}$, and it is easy to see that $\log{t} = O(\log{\beta})$, hence the complexity is  $O(\frac{\log^2{\beta}}{\epsilon})$ w.h.p.
\end{proof}

\subsection{($2+\epsilon$)-approximation of APSP}

We showed how to get a $(1+\epsilon,\beta)$-approximation for APSP, next we discuss a $(2+\epsilon)$-approximation, which gives a better approximation for short paths. 
As we showed before, for long paths of length around $t=O(\frac{\beta}{\epsilon})$, we already have a $(1+\epsilon)$-approximation from the emulator, so we need to take care only of short paths of length at most $t$. To explain the intuition, we start by describing a simple $(3+\epsilon)$-approximation, and then explain how to improve the approximation. Assume we sample a random set $A$ of $\sqrt{n}$ vertices, then each vertex has a vertex from $A$ among its $k=\sqrt{n}\log{n}$ closest vertices w.h.p. Now each vertex learns its \ktnearest vertices, which are the $k$ closest vertices of distance at most $t$. For any vertex $u$, there are 2 cases, either the \ktnearest vertices to $u$ contain its entire $t$-neighbourhood, in which case, $u$ already knows all the distances to vertices at distance at most $t$, or there are at least $k$ vertices in the $t$-neighbourhood of $u$, in which case there is also a vertex from $A$ there, $p_A(u)$. Now, for a pair of vertices $u,v$ of distance at most $t$, if $v$ is in the $(k,t)$-closest vertices to $u$, we are done. Otherwise, $d(u,p_A(u)) \leq d(u,v)$, which gives $$d(u,p_A(u))+d(p_A(u),v) \leq d(u,v) + d(p_A(u),u)+d(u,v) \leq 3d(u,v).$$ Hence, if we compute the distance from $u$ to $v$ through $p_A(u)$ we get a $3$-approximation for the distance. To do so, we let all vertices learn approximate distances to all vertices in $A$ at distance at most $2t$, this can be implemented efficiently using the bounded hopset and source detection algorithms. Since we only approximate the distances to $A$, this results in a $(3+\epsilon)$-approximation.

Obtaining a better approximation of $(2+\epsilon)$ requires several changes to the algorithm and analysis. At a high-level the algorithm starts by dealing with paths that have at least one high-degree vertex of degree at least $\sqrt{n}\log{n}$, for such paths it is relatively easy to find a good approximation using hitting set arguments. Then, we are left with a sparser graph of size $\widetilde{O}(n^{3/2})$. In this graph, we want to implement an algorithm similar to the $(3+\epsilon)$-approximation described above, compute a random set of vertices $A$, compute the \ktnearest vertices and compute distances to close by sources in $A$. However, since now we work only in a sparse graph, we can afford computing distances to a larger set $A$ of size around $n^{3/4}$, which allows focusing on $k$ around $n^{1/4}$. We show that we can exploit this sparsity and get a better approximation in this case using matrix multiplication of 3 sparse matrices.
This approach generally follows \cite{DBLP:conf/podc/Censor-HillelDK19}, however following the algorithm described there directly adds some logarithmic factors when we multiply the 3 matrices, and to avoid them we need to sparsify the graph even further. We next describe the algorithm in detail, and show the following. 

\begin{restatable}{theorem}{APSPthm}
Let $0<\epsilon<1$, there is a randomized $(2+\epsilon)$-approximation algorithm for unweighted undirected APSP in the \clique model that takes $O(\frac{\log^2{\beta}}{\epsilon})$ rounds w.h.p, where $\beta = O(\frac{\log{\log{n}}}{\epsilon})^{\log{\log{n}}}.$
\end{restatable}

We will need a few additional tools from \cite{DBLP:conf/podc/Censor-HillelDK19, censor2019sparse}.

\subsubsection*{Additional tools}

\paragraph{Distance through a set.}

Another useful tool from \cite{DBLP:conf/podc/Censor-HillelDK19} allows to compute the distance between vertices through a set. In the \distthrough problem, we assume that each vertex $v$ has a set $W_v$ and distance estimates $\delta(v,w)$ and $\delta(w,v)$ for all $w \in W_v$. The task is for all vertices $v$ to compute distance estimates
$\min_{w \in W_v \cap W_u } \{ \delta(v,w) + \delta(w,u) \}$
for all other vertices $u \in V$. The following is shown in \cite{DBLP:conf/podc/Censor-HillelDK19}.

\begin{theorem}\label{thrm:distance-through}
The \distthrough problem can be solved in
\[O \biggl( \frac{\rho^{2/3}}{n^{1/3}} + 1 \biggr)\]
rounds in \clique, where $\rho = \sum_{v \in V}|W_v|/n$.
\end{theorem}

\paragraph{Sparse matrix multiplication.}

In \cite{censor2019sparse}, it is shown how to multiply sparse matrices efficiently. Here, $\rho_S$ is the density of matrix $S$ which is the average number of non-zero elements in a row. In the case of distance products this corresponds to the average degree in the graph, and the zero element is $\infty.$ See more about distance products in Section \ref{sec:nearest}.
The following is shown in \cite{censor2019sparse, DBLP:conf/podc/Censor-HillelDK19}.\footnote{In \cite{censor2019sparse, DBLP:conf/podc/Censor-HillelDK19}, the statements are written in a slightly different but equivalent form. In \cite{censor2019sparse}, instead of looking at the density of a matrix $S$, they look at the total number of non-zero elements in the matrix which is equal to $n \rho_S$. In \cite{DBLP:conf/podc/Censor-HillelDK19}, there is an additional parameter related to the density of the output matrix, here we assume that the output matrix can be a full matrix, and bound this number with the maximum possible density, $n$.}

\begin{theorem} \label{thrm:mm}
Given two $n \times n$ matrices $S$ and $T$, there is deterministic algorithm that computes the product $P = S \cdot T$ over a semiring in the \clique model, completing in
\[O \biggl(\frac{(\rho_S \rho_T)^{1/3}}{n^{1/3}} + 1 \biggr)\] rounds.
\end{theorem}

\subsubsection*{The algorithm}

We next explain how the algorithm computes distances $d(u,v)$ for different pairs of vertices. We denote by $t = \frac{2\beta}{\epsilon}$.
We start with discussing the case that $d(u,v) \geq t$, in all other cases we assume that $d(u,v) \leq t.$

\paragraph{Paths where $d(u,v) \geq t$.} Here we use the emulator to get $(1+\epsilon)$-approximation for the distances.

\begin{claim} \label{claim_long_paths}
Let $t = \frac{2\beta}{\epsilon}$. There is a randomized algorithm that computes $(1+\epsilon)$-approximation for the distances $d(u,v)$ for all pairs of vertices $u,v$ with $d(u,v) \geq t$. The algorithm takes $O(\frac{\log^2{\beta}}{\epsilon})$ rounds w.h.p, for $\beta=O(\frac{\log{\log{n}}}{\epsilon})^{\log{\log{n}}}$.
\end{claim}

\begin{proof}
We use Theorem \ref{thm_emulator} to build $(1+\frac{\epsilon}{2}, \beta)$-emulator of size $O(n \log{\log{n}})$ for $\beta=O(\frac{\log{\log{n}}}{\epsilon})^{\log{\log{n}}}$ in $O(\frac{\log^2{\beta}}{\epsilon})$ rounds w.h.p, and then let all vertices learn it in $O(\log{\log{n}})$ rounds, as described in the proof of Theorem \ref{thm_near_additive_apsp}. 
Since $d(u,v) \geq t = \frac{2\beta}{\epsilon}$, the $(1+\frac{\epsilon}{2},\beta)$-emulator has a path between $u$ and $v$ of distance at most $(1+\frac{\epsilon}{2})d(u,v)+\beta \leq (1+\frac{\epsilon}{2})d(u,v) + \frac{\epsilon}{2} d(u,v) \leq (1+\epsilon)d(u,v).$
\end{proof}

\paragraph{Paths with a high-degree vertex and $d(u,v) \leq t$.}
We next focus on shortest paths $u-v$ that have at least one vertex with degree at least $\sqrt{n}\log{n}.$
We start by describing the algorithm, and then analyze its correctness and complexity.
In the algorithm, each pair of vertices $u,v$ have an estimate for the distance $\delta(u,v)$ between them, which is updated each time they learn about a shorter path between them. Let $N(v)$ be the set of neighbours of the vertex $v$. During the algorithm we build a hitting set $S$ of the sets $N(v)$ for all vertices where $|N(v)| \geq \sqrt{n} \log{n}$, this means that $S$ has at least one vertex in each of these sets.

\begin{oframed}
\begin{enumerate}
\item Set $\delta(u,v)=1$ if $\{u,v\} \in E$, and set $\delta(u,v)$ to the distance between $u$ and $v$ in the emulator otherwise.\label{alg_init}
\item Compute a hitting set $S$ of size $O(\sqrt{n})$ of the sets $N(v)$ for all vertices $v$ where $|N(v)| \geq \sqrt{n} \log{n}.$
\item Compute $(1+\frac{\epsilon}{2})$-approximate distances between all vertices to vertices of $S$ at distance at most $2t$, and for all $u \in V, s \in S$ update $\delta(u,s)$ accordingly.
\item For any pair of vertices $u,v$, set $\delta(u,v)= \min\{\delta(u,v), \min_{s \in S} \{\delta(u,s)+\delta(s,v)\} \}.$ \label{alg_dist}
\end{enumerate}
\end{oframed}

\begin{claim}
At the end of the algorithm, for any pair of vertices $u,v$ where $d(u,v) \leq t$ and the shortest $u-v$ path has at least one vertex of degree at least $\sqrt{n} \log{n}$, we have $\delta(u,v) \leq (2+\epsilon)d(u,v)$ w.h.p.
\end{claim}

\begin{proof}
For the analysis, we assume that $S$ is indeed a hitting set, which happens w.h.p if we use a randomized algorithm to construct $S$.
In this case, each vertex of degree at least $\sqrt{n}\log{n}$ has a neighbour in $S$. Let $u,v$ be such that $d(u,v) \leq t$ and the shortest $u-v$ path has at least one vertex $w$ of degree at least $\sqrt{n}\log{n}.$ We know that in this case, $w$ has a neighbour $s \in S$. Also, since $w$ is in the shortest path between $u$ and $v$, then $d(u,s) \leq t+1$ and $d(v,s) \leq t+1$. Since we computed $(1+\frac{\epsilon}{2})$-approximations between all vertices to vertices in $S$ of distance $2t$, we have that $\delta(u,s),\delta(v,s)$ have $(1+\frac{\epsilon}{2})$-approximations to the distances $d(u,s),d(v,s)$, respectively. Also, $$d(u,s)+d(s,v) \leq d(u,w) + d(w,s) + d(s,w) + d(w,v) = d(u,v)+2 \leq 2d(u,v),$$ where in the last inequality we assume $d(u,v) \geq 2$, if the distance is 1, then $u$ and $v$ are neighbours and know the distance between them. Since we approximate the distances, we get a $2(1+ \frac{\epsilon}{2})$-approximation, which is a $(2+\epsilon)$-approximation. 
\end{proof}

\begin{claim} \label{alg_time}
The algorithm takes $O(\frac{\log^2{t}}{\epsilon})$ rounds w.h.p.
\end{claim}

\begin{proof}
All vertices already computed the emulator in the first case of the algorithm (see Claim \ref{claim_long_paths}), which allows implementing Line \ref{alg_init} without communication.
Computing the hitting set $S$ is done using Lemma \ref{rand_hit}. As $S$ is a hitting set for sets $N(v)$ of size at least $k \geq \sqrt{n} \log{n}$, $S$ has size $O(\sqrt{n})$ w.h.p. The construction is completely local, and requires only one round to update all vertices which vertices are in $S$. To compute distances from $S$ we use the randomized algorithm in Theorem \ref{hopset_thm} to build a $(\beta',\epsilon',t')$-hopset $H'$ with $t' = 2t, \epsilon' = \frac{\epsilon}{2}, \beta' = O(\frac{\log{t}}{\epsilon}).$ The construction takes $O(\frac{\log^2{t}}{\epsilon})$ rounds. In the graph $G \cup H'$ there are $\beta'$-hop paths with $(1+\epsilon')$-approximate distances to all pairs of vertices at distance at most $2t.$ Now we run the \sdk algorithm with the set of sources $S$, and $d= \beta'$, which guarantees that all vertices $u \in V, s \in S$ of distance at most $2t$ would learn a $(1+\epsilon')$-approximation to the distance between them. Since $S$ is of size $O(\sqrt{n})$ w.h.p, the complexity of the \sdk algorithm is $O(\beta')=O(\frac{\log{t}}{\epsilon})$ w.h.p by Theorem \ref{thrm:source-detection}. To compute the minimum distance between $u$ and $v$ through a vertex in $S$ in Line \ref{alg_dist}, we use the \distthrough algorithm with the sets $W_v = S$ for all vertices. Since $S$ is of size $O(\sqrt{n})$ w.h.p, the complexity is $O(1)$ w.h.p from Theorem \ref{thrm:distance-through}. The overall complexity is $O(\frac{\log^2{t}}{\epsilon})$ rounds w.h.p.
\end{proof}

\paragraph{Paths that contain only low-degree vertices and $d(u,v) \leq t$.}
From now on we focus only on the case that the shortest $u-v$ path contains only vertices of degree at most $\sqrt{n}\log{n}.$
For this, it is enough to work with a sparse subgraph $G'$ that has only edges incident to vertices of degree at most $\sqrt{n}\log{n}.$ In this graph we can afford to compute distances to a set of sources of size $O(n^{3/4}/\log{n})$, we next explain how to exploit it. 
We start by describing the general structure of the algorithm, and then analyze it. The stretch analysis also provides a more intuitive description of the algorithm.
The value of $\delta(u,v)$ at the beginning is set to the value obtained in the previous algorithm, and updated later if $u,v$ find a shorter path between them. 

\begin{oframed}
\begin{enumerate}
\item Let $G'$ be a subgraph of $G$ that contains only all the edges incident to vertices of degree at most $\sqrt{n}\log{n}.$ All the computations in the algorithm are in the graph $G'$.
\item Let $k = n^{1/4} \log^2{n}.$ Compute for each vertex $u$ the distances to the \ktnearest vertices, denote this set by $N_{k,t}(u)$, update $\delta(u,v)$ accordingly for all $u \in V, v \in N_{k,t}(u).$\label{alg2_nearest}
\item Set $\delta(u,v) = \min\{\delta(u,v), \min_{w \in N_{k,t}(u) \cap N_{k,t}(v)} \{\delta(u,w)+\delta(w,v)\}\}.$\label{alg2_w}
\item Compute a hitting set $A$ of size $O(n^{3/4}/\log{n})$ of the sets $N_{k,t}(v)$ for all vertices where $|N_{k,t}(v)|=k$.
\item Compute $(1+\epsilon')$-approximate distances between all vertices to vertices of $A$ at distance at most $2t$ for $\epsilon' = \frac{\epsilon}{2}$, and for all $u \in V, s \in A$ update $\delta(u,s)$ accordingly.\label{alg2_A}
\item If $N_{k,t}(u) \cap A \neq \emptyset$, denote by $p^t_A(u)$ the closest vertex to $u$ from the set $A$, at distance at most $t$.
\item Let $\delta(u,v) = \min\{\delta(u,v), \delta(u,p^t_A(u)) + \delta(p^t_A(u),v), \delta(v,p^t_A(v)) + \delta(p^t_A(v),u)\}$, where the expressions containing $p^t_A(u), p^t_A(v)$ are only calculated if these vertices exist.\label{alg2_p}  
\item Let $N(v)$ be the set of neighbours of $v$ in $G'$. Compute a hitting set $A'$ of size $O(\sqrt{n}\log^5{n})$ for the sets $N(v)$ for all vertices where $|N(v)| \geq \frac{n}{k^2} = \frac{\sqrt{n}}{\log^4{n}}.$
\item Compute $(1+\epsilon')$-approximate distances between all vertices to vertices of $A'$ at distance at most $2t$ for $\epsilon' = \frac{\epsilon}{2}$, and for all $u \in V, s \in A'$ update $\delta(u,s)$ accordingly.\label{alg2_dA'}
\item For each vertex $u$, define a set $A'_u$ of size at most $k$ as follows. For each $v \in N_{k,t}(u)$ that has a neighbour $w \in A'$ add one such neighbour $w$ to $A'_u.$\label{alg2_A'u} 
\item Let $\delta(u,v) = \min\{ \delta(u,v), \min_{w \in A'_u} \delta(u,w) + \delta(w,v)\}.$\label{alg2_A'}
\item Let $E''$ be all edges in $G'$ with at least one endpoint of degree at most $\frac{n}{k^2} = \frac{\sqrt{n}}{\log^4{n}}$. 
\item Let $\delta'(u,v) = \min\{ \delta(u,u')+\delta(u',v')+\delta(v',v) | u' \in N_{k,t}(u), v' \in N_{k,t}(v), \{u',v'\} \in E''\}.$\label{alg2_m}
\item Set $\delta(u,v) = \min\{\delta(u,v),\delta'(u,v)\}.$
\end{enumerate}
\end{oframed}

We next prove the correctness of the algorithm, and then analyze its complexity.

\begin{claim}
For all pairs of vertices $u,v$ where $d_{G}(u,v) \leq t$ and the shortest $u-v$ path contains only vertices of degree at most $\sqrt{n}\log{n}$, at the end of the algorithm we have $\delta(u,v) \leq (2+\epsilon)d(u,v)$ w.h.p.
\end{claim}

\begin{proof}
We use in the analysis $d(u,v)$ for the distance in $G'$, for vertices where the shortest $u-v$ path in $G$ contains only vertices of degree at most $\sqrt{n}\log{n}$, we have $d_{G'}(u,v) = d_G(u,v)$, as the path is contained in $G'$.
We assume for the analysis that $A$ and $A'$ are indeed hitting sets, which happens w.h.p if we construct them with a randomized algorithm.
The algorithm basically deals with different pairs $u,v$ in different lines. We next explain the different cases.
First, if $v$ is among the \ktnearest vertices to $u$, then $u$ learns the distance in Line \ref{alg2_nearest}, and we are done. 
Hence, we are left only with the case that $v \not \in N_{k,t}(u)$ and vice verse. Note that in this case, since $d(u,v) \leq t$, and $v \not \in N_{k,t}(u)$, there are at least $k$ vertices at distance at most $t$ from $u$, which means that $|N_{k,t}(u)|=k.$ As $A$ is a hitting set for these sets, we have that there is a vertex from $A$ among $N_{k,t}(u)$, which means that the vertex $p^t_A(u)$ exists. A symmetric argument shows that $p^t_A(v)$ exists (if we did not already learn the distance $d(u,v)$).
Let $\pi(u,v)$ be a shortest $u-v$ path. Note that we are in the case that $\pi(u,v)$ has only vertices of degree at most $\sqrt{n} \log{n}$, hence it is contained in $G'$. The analysis breaks down to several cases.

\textbf{Case 1: there is $w \in \pi(u,v) \cap N_{k,t}(u) \cap N_{k,t}(v).$}
In this case, we have that $$d(u,v) = min_{w \in N_{k,t}(u) \cap N_{k,t}(v)} \{\delta(u,w)+\delta(w,v)\},$$ hence we learn the distance in Line \ref{alg2_w}.

\textbf{Case 2: there is $w \in \pi(u,v) \setminus (N_{k,t}(u) \cup N_{k,t}(v)).$}
In this case, we show that approximating the distance from $u$ to $v$ with the distance of a path that goes through a vertex in $A$ in Line \ref{alg2_p} gives a good approximation.
Since $w \in \pi(u,v)$, it is of distance at most $d(u,v)/2$ from one of $u$ or $v$, assume w.l.o.g that $d(u,w) \leq d(u,v)/2.$ Since $p^t_A(u) \in N_{k,t}(u)$ and $w \not \in N_{k,t}(u)$, we have $d(u,p^t_A(u)) \leq d(u,w) \leq d(u,v)/2.$
Also, $d(v,p^t_A(u)) \leq d(v,u) + d(u,p^t_A(u)) \leq \frac{3}{2} d(u,v) \leq \frac{3}{2} t.$ Since all vertices compute distances to all vertices in $A$ of distance at most $2t$ in Line \ref{alg2_A}, we have that $v$ computed $(1+\epsilon')$-approximate distance to $p^t_A(u).$ If we compute the distance from $u$ to $v$ through $p^t_A(u)$ we get 
$$d(u,p^t_A(u)) + d(p^t_A(u),v) \leq \frac{1}{2} d(u,v) + \frac{3}{2} d(u,v) = 2d(u,v).$$
Since we computed $(1+\epsilon')$-approximations to the values $d(u,p^t_A(u)),d(p^t_A(u),v)$, in Line \ref{alg2_p} we get a $(2+2\epsilon')=(2+\epsilon)$-approximation for $d(u,v)$.

\textbf{Case 3: all vertices in $\pi(u,v)$ are in exactly one of $N_{k,t}(u),N_{k,t}(v)$.}
In this case, the shortest path from $u$ to $v$ starts with a path from $u$ to $u' \in N_{k,t}(u)$, then an edge from $u'$ to a vertex $v' \in N_{k,t}(v)$ and a path from $v'$ to $v$. Ideally, we would like to exploit the simple structure of the path, in order to learn about its distance, a similar idea is used in \cite{DBLP:conf/podc/Censor-HillelDK19}. However, computing the distance directly would add logarithmic factors to the complexity. To avoid them, we break into cases according to the degree of $u'$. If the degree is large, we show that we can use hitting set arguments to approximate $d(u,v)$, and otherwise, we can exploit the fact that the degree is small to compute the exact distance $d(u,v).$

If $u'$ has degree at least $\frac{n}{k^2} = \frac{\sqrt{n}}{\log^4{n}}$, then $u'$ has a neighbour in $A'$, as $A'$ is a hitting set for vertices of degree at least $\frac{n}{k^2}$. 
Also, recall that for a vertex $u$, we define in Line \ref{alg2_A'u} the set $A'_u$ of size at most $k$, as follows. We look at all vertices in $N_{k,t}(u)$, and for each one of them that has a neighbouring vertex in $A'$, we add one such vertex to $A'_u$.  
In particular, since $u' \in N_{k,t}(u)$ and $u'$ has a neighbour in $A'$, we have that $A'_u$ has a vertex $w' \in A'$ which is a neighbour of $u'$. Also, in Line \ref{alg2_dA'}, all vertices computed $(1+\epsilon')$-approximate distances to vertices in $A'$ of distance at most $2t$, which includes distances from $u$ to $w'$ and from $v$ to $w'$ as $d(u,v) \leq t$ and $w'$ is adjacent to $\pi(u,v)$. Hence, we have
$$\min_{w \in A_u'} \{ \delta(u,w)+\delta(w,v) \} \leq (1+\epsilon') (d(u,w')+d(w',v)) \leq (1+\epsilon') (d(u,u')+d(u',w')+d(w',u')+d(u',v)).$$ 
This equals to $(1+ \epsilon')(d(u,v)+2) \leq (2+\epsilon)d(u,v),$
where we used the fact that $u'$ is in $\pi(u,v)$, $w'$ is a neighbour of $u'$, $d(u,v) \geq 2$ and $\epsilon'=\frac{\epsilon}{2}.$ 
Hence, $u$ and $v$ compute a $(2+\epsilon)$-approximation for the distance in Line \ref{alg2_A'}.

The only case left is that $u'$ has degree at most $\frac{n}{k^2}$. Hence, $d(u,v)=d(u,u')+d(u',v')+d(v',v)$, where $u' \in N_{k,t}(u), v' \in N_{k,t}(v), \{u',v'\} \in E''.$ Since in this case we also have $\delta(u,u')=d(u,u'),\delta(u',v')=d(u',v'),\delta(v',v)=d(v',v)$ as vertices know the exact distances to the \ktnearest vertices and neighbours, we have that $u,v$ compute $d(u,v)$ in Line \ref{alg2_m}. 
\end{proof}

\begin{claim}
The algorithm takes $O(\frac{\log^2{t}}{\epsilon})$ time w.h.p.
\end{claim}

\begin{proof}
Computing the distances to the \ktnearest to each vertex takes $O(\log^2{t})$ rounds by Theorem \ref{thrm:knearest} as $k = n^{1/4} \log^2{n} = O(n^{2/3}).$
Computing $\min_{w \in N_{k,t}(u) \cap N_{k,t}(v)} \{\delta(u,w)+\delta(w,v)\}$ is done using the \distthrough algorithm with the sets $W_v = N_{k,t}(v)$. Since these sets are of size at most $k < \sqrt{n}$, the complexity is $O(1)$ from Theorem \ref{thrm:distance-through}.

Computing $A$ is done using Lemma \ref{rand_hit} which requires only one round to update all vertices which vertices are in $A$. 
To compute the approximate distances to vertices in $A$ we use the bounded hopset and \sdk algorithms as described in the proof of Claim \ref{alg_time}. The only difference here is the size of the graph and the size of $A$. Here, we compute distances in the graph $G' \cup H'$ where $H'$ is the $(\beta',\epsilon',2t)$-hopset constructed, the number of edges in $G'$ and in $H'$ is bounded by $O(n^{3/2} \log{n})$ w.h.p, from the definition of $G'$ and from Theorem \ref{hopset_thm}. Also, $A$ is of size $O(n^{3/4}/\log{n})$ w.h.p by its definition. As $G' \cup H'$ has $\beta'$-hop paths that approximate all paths of length at most $2t$ in $G'$, for $\beta'=O(\frac{\log{t}}{\epsilon})$, the complexity of the \sdk algorithm is  $O \biggl(\biggl( \frac{ (n^{3/2}\log{n})^{1/3} (n^{3/4}/\log{n})^{2/3}}{n} +1 \biggr) \beta' \biggl) = O(\beta')=O(\frac{\log{t}}{\epsilon})$ w.h.p. As constructing the hopset $H'$ requires $O(\frac{\log^2{t}}{\epsilon})$, this is the total complexity for this part.

Each vertex $v$ can compute locally the vertex $p^t_A(v)$ from the set $N_{k,t}(v)$, and then it sends to all vertices $p^t_A(v)$ if exists. Hence, all pairs of vertices $u,v$ can compute simultaneously the values $\delta(u,p^t_A(u)) + \delta(p^t_A(u),v), \delta(v,p^t_A(v)) + \delta(p^t_A(v),u)$, by exchanging between them the relevant values.

The computation of $A'$ is done using Lemma \ref{rand_hit} and the computation of the approximate distances to $A'$ follows the description of the same computation for $A$. $A'$ is smaller which can only improve the complexity. Each $u$ computes the set $A'_u$ as follows. $u$ already computed the set $N_{k,t}(u)$. Also, each vertex that has at least one neighbour in $A'$ can inform all vertices about such neighbour in one round. The set $A'_u$ is then composed of all such neighbours of vertices in $N_{k,t}(u).$ We next explain how to compute $\min_{w \in A'_u} \{\delta(u,w)+\delta(w,v)\}$ in Line \ref{alg2_A'}. For this, we use sparse matrix multiplication with matrices $M_1,M_2$ where $M_1[u,v] = \delta(u,v)$ if $v \in A'_u$ and it equals $\infty$ otherwise, and the matrix $M_2[u,v]=\delta(u,v)$ if $u \in A'$, and it equals $\infty$ otherwise. Note that $M_2$ has all distances from $A'$ to other vertices, where $M_1$ only has distances from each $u$ to $A'_u.$ Now by definition of matrix multiplication in the min-plus semiring $(M_1 \cdot M_2)[u,v]=\min_{w \in A'_u} \{\delta(u,w)+\delta(w,v)\}$, as needed. The density of $M_1$ is bounded by $k=n^{1/4}\log^2{n}$ as the sets $A'_u$ are of size at most $k$ by definition, and the density of $M_2$ is bounded by $|A'|=\sqrt{n}\log^5{n}$. Hence, by Theorem \ref{thrm:mm} the computation takes $O(1)$ rounds.

Finally, we explain how to compute $$\delta'(u,v) = \min\{ \delta(u,u')+\delta(u',v')+\delta(v',v) | u' \in N_{k,t}(u), v' \in N_{k,t}(v), \{u',v'\} \in E''\}.$$ For this, we use sparse matrix multiplication with the following 3 matrices. One, $W_1$, has the distances from each vertex $u$ to the set $N_{k,t}(u)$. The second, $W_2$, has all edges from vertices of degree at most $\frac{n}{k^2}$ in $G'$ to all their neighbours, and the third $W_3 = W_1^T$, i.e., for each vertex $v$, $W_3[v',v]$ has the distance $d(v',v)$ if $v' \in N_{k,t}(v)$, and it has $\infty$ otherwise. Let $W$ be the product of these 3 matrices in the min-plus semiring. By definition of the product, $W[u,v]$ is the weight of the minimum path that has the first edge from $u$ to $u' \in N_{k,t}(u)$, the second edge from $u'$ to $v'$ where $u'$ has degree at most $\frac{n}{k^2}$ and $v'$ is a neighbour of $u'$ in $G'$, and the third edge is from $v' \in N_{k,t}(v)$ to $v$. Let $P$ be a shortest path between $u$ to $v$ that is composed of a path between $u$ to $u' \in N_{k,t}(u)$, an edge $\{u',v'\} \in E''$, and a path from $v' \in N_{k,t}(v)$ to $v$, by definition of $E''$ we have that at least one of $u',v'$ has degree at most $\frac{n}{k^2}$. If this is $u'$, then we have $W[u,v]=\delta'(u,v)$, and otherwise we have $W[v,u]=\delta'(v,u)=\delta'(u,v)$. Hence, at least one of $u,v$ learns the value $\delta'(u,v)$ and by taking the minimum value between $W[u,v],W[v,u]$ both learn it.
We now analyze the complexity. In $W_1$, the degrees of all vertices are at most $k=n^{1/4}\log^2{n}$ by definition, hence its density is $k$. In $W_2$ the degrees of all vertices are most $\frac{n}{k^2}$, hence its density is $\frac{n}{k^2}$. Hence, we can multiply $W_1,W_2$ in $O(1)$ rounds by Theorem \ref{thrm:mm}. Moreover, in the product $W_1 \cdot W_2$ the degrees of all vertices are at most $k \cdot \frac{n}{k^2}=\frac{n}{k}$. Since $W_3 = W_1^T$, its density is at most $k$. Hence, multiplying $W_1 \cdot W_2$ and $W_3$ takes $O(1)$ rounds as well by Theorem \ref{thrm:mm}, as the product of densities is at most $\frac{n}{k} \cdot k=n$.
\end{proof}

At the end, each pair of vertices $u,v$ have an estimate $\delta(u,v)$ for the distance between them which is the minimum estimate obtained in the different parts of the algorithm. From the discussion above, in all cases we have that $\delta(u,v) \leq (2+\epsilon)d(u,v)$ w.h.p. Also, the total time complexity is $O(\frac{\log^2{\beta}}{\epsilon}+\frac{\log^2{t}}{\epsilon})$ for $\beta=O(\frac{\log{\log{n}}}{\epsilon})^{\log{\log{n}}}$ and $t = \frac{2\beta}{\epsilon}$. By the choice of parameters, this gives a complexity of $O(\frac{\log^2{\beta}}{\epsilon})$ rounds, which gives the following.

\APSPthm*

\section{Deterministic Algorithms}

We next explain how to derandomize our algorithms. Most of the randomized parts in our algorithms are based on hitting set arguments, which can be derandomized easily by using Lemma \ref{det_hit} instead of Lemma \ref{rand_hit}. This adds $O((\log{\log{n}})^3)$ term to the complexity. However, derandomizing the emulator algorithm requires a more careful process. Intuitively, using hitting set arguments directly to derandomize the process would add logarithmic factor to the size of the emulator, which leads to additional logarithmic term in the complexity of our shortest paths algorithms, which is too expensive. To avoid it, we introduce the \emph{soft hitting set} problem, a new variant of the hitting set problem that captures more accurately the randomized process required in the emulator construction. Then, we show how to build soft hitting sets deterministically, which leads to deterministic construction of emulators with essentially  the same parameters as the randomized construction.
Other than this, another slight difference in the deterministic construction is the following. In our randomized construction we used the set $S_r$ of the emulator in two different roles: First, it is one of the sets of the emulator. Second, it is a hitting set for heavy vertices (see Claim \ref{claim_heavy}). In the deterministic construction we add a different hitting set $A$ for the second purpose, which slightly changes the description of the algorithm.

We next describe in detail the deterministic construction of emulators, and show deterministic variants of all our applications. Later, in Section \ref{sec:derandomization}, we show how to build soft hitting sets deterministically.

\subsection{Deterministic emulators}

For the construction, we need the following definitions.
For a given two subsets of vertices $V_1,V_2$, define the \emph{soft-hitting set} function $SH(V_1,V_2)$ by 
$$
SH(V_1,V_2)=
\begin{cases}
0, \text{~~if~~} V_1 \cap V_2 \neq \emptyset\\
|V_1|, \text{~~otherwise.}\\
\end{cases}
$$

\begin{definition}[Soft Hitting Set]
Consider a graph $G=(V,E)$ with two special sets of vertices $L \subseteq V$ and $R \subseteq V$ with the following properties: each vertex $u \in L$ has a subset $S_u \subseteq R$ where $|S_u| \geq \Delta$.
A set of vertices $R^* \subseteq R$ is \emph{soft hitting set} for $L$ if:
(i) $|R^*|=O(|R|/\Delta)$ and 
(ii) $\sum_{u \in L} SH(S_u, R^*)=O(|L|\cdot \Delta)$. 
\end{definition}
The above definition can be viewed as an adaption of the hitting-set definition by Ghaffari and Kuhn in \cite{ghaffari2018derandomizing}.  In Section \ref{sec:derandomization} we show:
\begin{restatable}{lemma}{hitting}[Det. Construction of Soft Hitting Sets] \label{lemma_soft_hitting}
Let $L,R \subseteq V$ be subsets of vertices where each vertex $u \in L$ knows a set $S_u \subseteq R$ of at least $\Delta$ vertices. There exists an 
$O((\log\log n)^3)$-round deterministic algorithm in the \clique model that computes a soft-hitting set $R^* \subseteq R$ for $L$, where $|R^*| \leq c|R|/\Delta$ for a constant $c$.
\end{restatable}

We next show how to use soft hitting sets to derandomize our emulator algorithm.
The main difference in the algorithm is the construction of the sets $S_i$, before we just defined $S_i \gets Sample(S_{i-1}, p_i)$, here we use soft hitting sets to build them.

\subsubsection*{Constructing the sets $S_i$}

First, we build sets $\emptyset = S'_{r+1} \subset S'_r \subset S'_{r-1} \ldots \subset \ldots S'_1 \subset S'_0=V$, which behave similarly to the sets $S_i$ in the randomized construction.
The values $\delta_i$ and $p_i$ are defined as in the randomized construction.
Before the process we run the \ktnearest algorithm with $k = n^{2/3}, t = \delta_r$ to let each vertex $v$ learn the set $N_{k,t}(v)$ of $k$ closest vertices of distance at most $t$, the complexity is $O(\log^2{\delta_r})$ rounds by Theorem \ref{thrm:knearest}. 

We build the sets $S'_i$ in an iterative process that takes $r$ iterations.
For $0 \leq i \leq r-1$, given the set $S'_i$, we define the set $S'_{i+1}$ as follows.
First, for each vertex $v \in S'_i$, we say that $v$ is \emph{heavy} if $|B(v,\delta_i,G)| \geq n^{2/3}$, and it is \emph{light} otherwise.
We define the set $L$ to be composed of all light vertices $v \in S'_i$ where $|B(v,\delta_i,G) \cap S'_i| \geq \frac{c}{p_{i+1}}$, where $c$ is the constant guaranteed by Lemma \ref{lemma_soft_hitting}. For a vertex $v \in L$, we define $T_v = B(v,\delta_i,G) \cap S'_i$, and we define $R = S'_i$ and $\Delta = \frac{c}{p_{i+1}}$. By definition, for all $v \in L$ we have $T_v \subseteq R$, and $|T_v| \geq \Delta.$ Additionally, given $S'_i$, all vertices $v$ can learn if they are in $L$, and if so the corresponding set $T_v$, as follows. All vertices computed the sets $N_{k,t}(v)$ for $k = n^{2/3}, t = \delta_r \geq \delta_i.$ If $|B(v,\delta_i,G)| < n^{2/3}$, then it is contained entirely in $N_{k,t}(v)$, $v$ can learn this by checking if the set $N_{k,t}(v)$ has a vertex of distance greater than $\delta_i$ or not, hence it can learn if it is light. In the case it is indeed light, we have that $B(v,\delta_i,G) \cap S'_i \subset N_{k,t}(v)$, hence it can compute this set without communication. Note that by definition this set is $T_v$ if $v \in L$, and in addition $v$ can deduce if it is in $L$ based on the size of this set. Hence, all the conditions of Lemma \ref{lemma_soft_hitting} hold, and we can use it to compute a set $S'_{i+1} \subseteq S'_i$ which is a soft hitting set with respect to the sets $\{T_v\}_{v \in L}.$ The complexity is $O((\log{\log{n}})^3)$ rounds per iteration. At the end, we have additional round to let all vertices learn which vertices are in $S'_{i+1}.$
This completes the description of one iteration. Since we have $r$ iterations, the overall complexity is $O(r(\log{\log{n}})^3).$

Other than the sets $S'_i$ we also compute a set $A$ which is a hitting set for heavy vertices. $A$ is defined as follows. For a vertex $v$ which is heavy in some iteration, let $i$ be the first iteration that $v \in S'_i$ is heavy. Define $A_v = N_{k,\delta_{i}}(v)$, note that $v$ knows this set as it is contained in $N_{k,\delta_r}(v)$, and that $|A_v|=k$ since $v$ is heavy in iteration $i$. The set $A$ would be a hitting set of the sets $A_v$ of heavy vertices. As $k = n^{2/3}$, we can construct deterministically a hitting set $A$ of size $O(n^{1/3} \log{n})$ that hits all sets $A_v$ in $O((\log{\log{n}})^3)$ rounds using Lemma \ref{det_hit}. 

Finally, we define the set $S_i = S'_i \cup A$ for all $0 \leq i \leq r$, and $S_{r+1} = \emptyset$. By their definition, we have that $\emptyset = S'_{r+1} \subset S'_r \subset S'_{r-1} \ldots \subset \ldots S'_1 \subset S'_0=V$.
From now on the algorithm works as the randomized one with respect to the sets $S_i$. For example, a vertex $v \in S_i$ is $i$-dense if $B(v,\delta_i,G) \cap S_{i+1} \neq \emptyset,$ the definition of $c_i(v)$ is with respect to the new $S_i$'s, and so on.
Most of the algorithm and analysis, including the stretch analysis, do not depend on the specific way the sets $S_i$ are constructed, hence they remain the same. The differences are mostly in the size analysis, and also slightly in the implementation details. We next discuss the changes in detail.
From the discussion above we have the following.

\begin{claim} \label{construct_Si_det}
The construction of the sets $S_i$ takes $O(r (\log\log{n})^3 + \log^2{\delta_r})$ rounds.
\end{claim}

\begin{proof}
Computing the sets $N_{k,t}(v)$ for $t = \delta_r$ takes $O(\log^2{\delta_r})$ rounds from Theorem \ref{thrm:knearest}. The construction of the sets $S'_i$ takes $r$ iterations of $O((\log{\log{n}})^3)$ rounds by Lemma \ref{lemma_soft_hitting}, and the construction of $A$ takes $O((\log{\log{n}})^3)$ rounds by Lemma \ref{det_hit}.
\end{proof}

\subsubsection*{Size analysis}

We have the following bound on the size of the sets $S_i.$

\begin{claim} \label{claim_det_size}
For $0 \leq i <r$, it holds that $S_i$ is of size $O(n^{1-\frac{2^i-1}{2^r}})$, and $S_r$ is of size $O(\sqrt{n})$.
\end{claim}

\begin{proof}
First, we prove that for $0 \leq i< r$, the set $S'_i$ is of size at most $n^{1-\frac{2^i-1}{2^r}}.$
For $i=0$, it holds since $S_0 = V.$
For $1 \leq i < r$, the proof is by induction.
By definition, $S'_{i+1}$ is a soft hitting set, where we have $R = S'_i$ and $\Delta = \frac{c}{p_{i+1}}$. From Lemma \ref{lemma_soft_hitting}, we have that $|S'_{i+1}| \leq c|R|/\Delta = c|S'_i|p_{i+1}/c = |S'_{i}|p_{i+1}$. From the induction hypothesis and since $p_{i+1} = n^{- \frac{2^i}{2^r}}$, we have, $$|S'_{i+1}| \leq n^{1-\frac{2^i-1}{2^r}} \cdot n^{- \frac{2^i}{2^r}} = n^{1-\frac{2^{i+1}-1}{2^r}}.$$
For the case $i=r$, since $p_r = n^{-\frac{1}{2^r}}$, following the same argument we have that $$|S'_r| \leq |S'_{r-1}| p_r = n^{1-\frac{2^{r-1}-1}{2^r}} \cdot n^{-\frac{1}{2^r}} = \sqrt{n}.$$

Now, we have $S_i = S'_i \cup A$, where $A$ is of size $O(n^{1/3} \log{n})$. As $n^{1-\frac{2^i-1}{2^r}} > n^{1/3} \log{n}$ for $i <r$, and $|S'_r| \leq \sqrt{n}$, we have that $|S_i| = O(n^{1-\frac{2^i-1}{2^r}})$ for $0 \leq i < r$, and $S_r = O(\sqrt{n})$.
\end{proof}

We next bound the size of the emulator.

\begin{claim}
The emulator has $O(r \cdot n^{1+\frac{1}{2^r}})$ edges.
\end{claim}

\begin{proof}
We show that in each iteration the number of edges added to the emulator is $O(n^{1+\frac{1}{2^r}}).$ Since we have $r$ iterations, the claim follows.
First, for $i=r$, since $S_r = O(\sqrt{n})$ by Claim \ref{claim_det_size}, even if all vertices of $S_r$ add edges to all vertices of $S_r$ this adds a total of $O(n)$ edges. We next fix $0 \leq i < r$ and bound the edges added by vertices in $S_i \setminus S_{i+1}$. Let $v \in S_i \setminus S_{i+1}$. First note that each $i$-dense vertex $v \in S_i \setminus S_{i+1}$ only adds one edge to the emulator by the definition of the algorithm. Also, by definition, $v$ is $i$-dense if $B(v, \delta_i, G) \cap S_{i+1} \neq \emptyset$. Since $A \subseteq S_{i+1}$, we have that if  $B(v, \delta_i, G) \cap A \neq \emptyset$, then $v$ is $i$-dense. We next divide to cases and see how many edges are added by $v$.

\emph{\textbf{Case 1:} $v$ is heavy in iteration $i$.} In this case, from the definition of $A$, we have that $B(v, \delta_{i'},G) \cap A \neq \emptyset$ where $i' \leq i$ is the first iteration where $v$ is heavy. Since $i' \leq i$, we have that $\delta_{i'} \leq \delta_i$ and $B(v,\delta_{i'},G) \subseteq B(v,\delta_i,G).$ It follows that $B(v, \delta_{i},G) \cap A \neq \emptyset$, which means that $v$ is $i$-dense and adds only one edge, as explained above.

\emph{\textbf{Case 2:} $B(v, \delta_i,G) \cap A \neq \emptyset$.} As explained above, in this case, $v$ is $i$-dense and adds only one edge to the emulator.

\emph{\textbf{Case 3:} $|B(v, \delta_i,G) \cap S_i| \leq \frac{c}{p_{i+1}}$.} As in all cases, $v$ only adds edges to vertices in $B(v, \delta_i,G) \cap S_i$, it adds at most $\frac{c}{p_{i+1}}$ edges in this case.

\emph{\textbf{Case 4:} $v$ is light in iteration $i$, $B(v, \delta_i,G) \cap A = \emptyset$ and $|B(v, \delta_i,G) \cap S_i| \geq \frac{c}{p_{i+1}}$.} First, since $S_i = S'_i \cup A$ and $B(v, \delta_i,G) \cap A = \emptyset$, we have that $|B(v, \delta_i,G) \cap S'_i| = |B(v, \delta_i,G) \cap S_i| \geq \frac{c}{p_{i+1}}$. In addition, we are in the case that $v$ is light in iteration $i$. We have that $v \in S_i \setminus S_{i+1}$, since $S_i = S'_i \cup A$, and $A \subseteq S_{i+1}$, it follows that $v \in S'_i$. Hence, by the definition of the set $L$, we have that $v \in L$, and $T_v = B(v, \delta_i,G) \cap S'_i$. There are two options for $v$, either $T_v \cap S'_{i+1} \neq \emptyset$, in this case, $v$ is $i$-dense and adds only one edge to the emulator, or $T_v \cap S'_{i+1} = \emptyset$, in this case $v$ is $i$-sparse, in which case it adds $|T_v|$ edges to the emulator. We denote by $L'$ all the vertices in $L$ which are in this second case and add $|T_v|$ edges. Also, we denote by $L'' \subseteq L$ all vertices in $L$ which are in case 4. Recall that by the definition of the soft hitting set function we have that $SH(T_v, S'_{i+1}) = 0$ if $T_v \cap S'_{i+1} \neq \emptyset$, and it equals $|T_v|$ otherwise. Hence, we have the following. All vertices in Case 4 that are not in $L'$ are $i$-dense and only add one edge to the emulator. The total number of edges added by vertices in $L'$ is $$\sum_{v \in L'} |T_v| = \sum_{v \in L''} SH(T_v, S'_{i+1}) \leq \sum_{v \in L} SH(T_v, S'_{i+1}) = O(|L| \cdot \Delta),$$ where the last equality follows from the fact that $S'_{i+1}$ is a soft hitting set of the sets $\{T_v\}_{v \in L}$. We have that $L \subseteq S_i$, and $\Delta = \frac{c}{p_{i+1}}$, which gives $\sum_{v \in L'} |T_v| = O(\frac{|S_i|}{p_{i+1}})$.

To sum up over all cases: the total number of edges added by $i$-dense vertices in bounded by $O(|S_i|)$, the total number of edges added by vertices in Cases 3 and 4 is bounded by $O(\frac{|S_i|}{p_{i+1}})$, hence the total number of edges added is bounded by $O(\frac{|S_i|}{p_{i+1}})$.
From Claim \ref{claim_det_size} we have $|S_i| = O(n^{1-\frac{2^i-1}{2^r}})$, and by the definition of $p_i$, we have that $\frac{1}{p_{i+1}} = n^{\frac{2^i}{2^r}}$ for $0 \leq i < r-1$, and $\frac{1}{p_r} = n^{\frac{1}{2^r}} < n^{\frac{2^{r-1}}{2^r}}$ for $i = r-1$. This gives $O(\frac{|S_i|}{p_{i+1}}) = O(n^{1-\frac{2^i-1}{2^r}} \cdot n^{\frac{2^i}{2^r}}) = O(n^{1+\frac{1}{2^r}}).$ This completes the proof.
\end{proof}

\subsubsection*{Implementation details}

The main change in the algorithm is the construction of the sets $S_i$ which is already discussed above. Adding edges to the emulator basically follows the ideas discussed in the randomized construction with slight changes as the definition of heavy and light vertices slightly changed, for completeness we next describe this part.

\begin{claim} \label{claim_edges_not_Sr_det}
All the edges of the emulator with at least one endpoint in $V \setminus S_r$ can be added to the emulator with correct distances in $O(\log^2{\delta_r})$ rounds. 
\end{claim}

\begin{proof}
We compute the set $N_{k,d}(v)$ of the \kdnearest for $k=n^{2/3}, d=\delta_r$, which takes $O(\log^2{\delta_r})$ time by Theorem \ref{thrm:knearest}. 
Let $v \in S_i \setminus S_{i+1}$ for $0 \leq i < r$, note that in particular $v \not \in A$, as $A \subseteq S_{i+1}$, hence $v \in S'_i$. If $v$ is light in iteration $i$, we have that $|B(v, \delta_i,G)| < n^{2/3}$. In this case $B(v, \delta_i,G) \subseteq N_{k,d}(v)$, as $\delta_i \leq d$. This means that $v$ knows the whole set $B(v, \delta_i,G)$ and can add edges accordingly: if this set contains a vertex in $S_{i+1}$ it adds one edge to it, and otherwise it adds edges to all vertices in $B(v, \delta_i,G) \cap S_{i}.$

The other case is that $v$ is heavy in iteration $i$, in this case from the definition of $A$ it follows that $N_{k,\delta_{i'}}(v) \cap A \neq \emptyset$, where $i'$ is the first iteration where $v$ is heavy. Since $i' \leq i$, we have $N_{k,\delta_{i'}}(v) \subseteq N_{k,\delta_{i}}(v) \subseteq N_{k,\delta_{r}}(v).$ Since $A \subseteq S_{i+1}$ it follows that $v$ is $i$-dense in this case and only adds one edge to the emulator to the closest vertex in $S_{i+1}$. Also, since $v$ knows the $k$ closest vertices, it can find this vertex and add an edge to it.
The whole complexity of the algorithm is $O(\log^2{\delta_r})$ time for computing the \kdnearest.
\end{proof}

Adding edges between vertices in $S_r$ follows the proof of Claim \ref{claim_edges_Sr} with the following changes: the construction of hopsets is done using the deterministic variant in Theorem \ref{hopset_thm}, which adds $O((\log{\log{n}})^3)$ term to the complexity, and the size of $S_r$ now is $O(\sqrt{n})$ always and not just w.h.p, which gives the following.

\begin{claim} \label{claim_edges_Sr_det}
All the edges in the emulator with two endpoints in $S_r$ can be added to the emulator with $(1+\epsilon')$-approximate distances in $O(\frac{\log^2{\delta_r}}{\epsilon'} + (\log{\log{n}})^3)$ rounds.
\end{claim}

From Claims \ref{construct_Si_det}, \ref{claim_edges_not_Sr_det} and \ref{claim_edges_Sr_det}, we have the following.

\begin{lemma} 
The time complexity of the algorithm is $O(\frac{\log^2{\delta_r}}{\epsilon'} + r (\log\log{n})^3)$ rounds.
\end{lemma}

The rest of the analysis and conclusion follow the randomized construction, which gives the following. For a constant $\epsilon$ and $r = \log{\log{n}}$, we get a complexity of $O((\log{\log{n}})^4).$

\begin{theorem}
Let $G$ be an unweighted undirected graph, let $0<\epsilon<1$ and let $r \geq 2$ be an integer, there is a deterministic algorithm that builds an emulator $H$ with $O(r \cdot n^{1+\frac{1}{2^r}})$ edges, and stretch of $(1+\epsilon, \beta),$ in $O(\frac{\log^2{\beta}}{\epsilon} + r (\log\log{n})^3)$ rounds, where $\beta = O(\frac{r}{\epsilon})^{r-1}$. For the choice $r = \log{\log{n}}$, we have $O(n \log{\log{n}})$ edges, a complexity of $O(\frac{\log^2{\beta}}{\epsilon} + (\log\log{n})^4)$ rounds, and $\beta = O(\frac{\log{\log{n}}}{\epsilon})^{\log{\log{n}}}.$
\end{theorem}

\subsection{Applications}

Using the deterministic emulator, we can get deterministic variants of all our applications.
First, by building the emulator and letting all vertices learn it, we have the following.

\begin{theorem} 
Let $0<\epsilon<1$, there is a deterministic $(1+\epsilon, \beta)$-approximation algorithm for unweighted undirected APSP in the \clique model that takes $O(\frac{\log^2{\beta}}{\epsilon} + (\log{\log{n}})^4)$ rounds, where $\beta = O(\frac{\log{\log{n}}}{\epsilon})^{\log{\log{n}}}.$
\end{theorem}

To get an approximation for multi-source shortest paths we follow the proof of Theorem \ref{thm_maap} with the only difference that we use the deterministic emulator and hopset constructions, which in total add an $O((\log{\log{n}})^4)$ term to the complexity.

\begin{theorem} 
Let $0<\epsilon<1$ and let $G$ be an unweighted undirected graph, there is a deterministic $(1+\epsilon)$-approximation algorithm in the \clique model for multi-source shortest paths from a set of sources $S$ of size $O(\sqrt{n})$ that takes $O(\frac{\log^2{\beta}}{\epsilon} + (\log{\log{n}})^4)$ rounds, where $\beta = O(\frac{\log{\log{n}}}{\epsilon})^{\log{\log{n}}}.$
\end{theorem}

To get a $(2+\epsilon)$-approximation for APSP we follow the randomized algorithm with the following changes: we use the deterministic emulator and hopset algorithms. In addition, we use the deterministic hitting set construction from Lemma \ref{det_hit} to build the hitting sets $S,A,A'$ in the algorithm. Note that these are hitting sets either of the sets $N(v)$ of neighbours of vertices, or of the sets $N_{k,t}(v)$ of the $k$ closest vertices at distance at most $t$, all these sets are known to the vertices during the algorithm which allows using Lemma \ref{det_hit}. The deterministic hitting set and hopset constructions add an additional $O((\log{\log{n}})^3)$ term and the deterministic emulator construction adds $O((\log{\log{n}})^4)$ term to the complexity, which gives the following.

\begin{theorem}
Let $0<\epsilon<1$, there is a deterministic $(2+\epsilon)$-approximation algorithm for unweighted undirected APSP in the \clique model that takes $O(\frac{\log^2{\beta}}{\epsilon} + (\log{\log{n}})^4)$ rounds, where $\beta = O(\frac{\log{\log{n}}}{\epsilon})^{\log{\log{n}}}.$
\end{theorem}

\subsection{Derandomization of soft hitting sets} \label{sec:derandomization}
We begin by setting up some notation, similarly as in \cite{ParterY18}. For a set $S$ we denote by $x \sim S$ a uniform sampling from $S$. For a function $\PRG$ and an index $i$, let $\PRG(s)_i$ 
the $\ith{i}$ bit of $\PRG(s)$. 

\begin{definition}[Pseudorandom Generators]
A generator $\PRG \colon \BB^r \to \BB^n$ is an $\epsilon$-pseudorandom
generator (PRG) for a class $\cC$ of Boolean functions if for every $f \in 
\cC$:
$$
|\Exp{x \sim \BB^n}{f(x)} - \Exp{s \sim \BB^r}{f(\PRG(s))}| \le 
\epsilon.
$$
We refer to $r$ as the seed-length of the generator and say $\PRG$ is explicit 
if there is an efficient algorithm to compute $\PRG$ that runs in time 
$poly(n, 1/\epsilon)$.
%We say $\PRG$ $\epsilon$-fools $\cC$ and refer to $\epsilon$ as the error.
\end{definition}

As in \cite{ParterY18}, we use the PRGs of \cite{gopalan2012better} which fool
families of DNFs using a short seed.

\begin{theorem}\label{thm:prg}
For every $\epsilon=\epsilon(n) > 0$, there exists an explicit pseudorandom 
generator, $\PRG 
\colon 	\BB^r \to \BB^n$ that fools all read-once DNFs on $n$-variables with 
error at most $\epsilon$ and seed-length $r = O((\log(n/\epsilon)) \cdot (\log 
\log(n/\epsilon))^3)$.
\end{theorem}

%\mtodo{I moved the definitions of hitting sets earlier, the lemma appears in both places.}

\remove{
For a given two subsets of vertices $V_1,V_2$, define the \emph{soft-hitting set} function $SH(V_1,V_2)$ by 
$$
SH(V_1,V_2)=
\begin{cases}
0, \text{~~if~~} V_1 \cap V_2 \neq \emptyset\\
|V_1|, \text{~~otherwise.}\\
\end{cases}
$$
\begin{definition}[Soft Hitting Set]
Consider a graph $G=(V,E)$ with two special sets of vertices $L \subseteq V$ and $R \subseteq V$ with the following properties: each vertex $u \in L$ has a subset $S_u \subseteq R$ where $|S_u| \geq \Delta$.
A set of vertices $R^* \subseteq R$ is \emph{soft hitting set} for $L$ if:
(i) $|R^*|=O(|R|/\Delta)$ and 
(ii) $\sum_{u \in L} SH(S_u, R)=O(|L|\cdot \Delta)$. 
\end{definition}
In this section we show:
\begin{lemma}[Det. Construction of Soft Hitting Sets]
Let $L,R \subseteq V$ be subsets of vertices where each vertex $u$ knows a set $S_u \subseteq R$ of at least $\Delta$ vertices. There exists an 
$O((\log\log n)^3)$-round deterministic algorithm that computes a soft-hitting set $R^* \subseteq R$ for $L$. 
\end{lemma}
}

In this section we prove Lemma \ref{lemma_soft_hitting} which states the following:
\hitting*

We first show that a soft hitting set can be obtained using a small random seed, and later explain how to compute it deterministically via derandomization in the \clique model. 
\begin{lemma}\label{lem:hittingrandloglog}[Soft Hitting Set with Small Seed]
Let $S$ be a subset of $[N]$ where $|S|\ge 
\Delta$ for some parameter $\Delta \le N$ and let $c$ be any constant. 
Then, there exists a family of hash functions $\cH= \{h \colon [N] \to \BB \}$ such 
that choosing a random function from $\cH$ takes $r=O(\log N \cdot (\log\log N)^3)$ 
random bits and for $Z_h = \{u \in [N] : h(u)=1 \}$ it holds that:
(i) $\prob{h}{ |Z_h| =O(N/\Delta)} \ge 2/3$, and
(ii) $\Exp{h}{ SH(S,Z_h) }=O(\Delta)$.
\end{lemma}
\begin{proof}
We first describe the construction of $\cH$. Let $p=c'/\Delta$ for some large constant $c'$ (will be set later), and let $\ell=\lfloor \log 1/p \rfloor$. Let 
$\PRG \colon \BB^r \to \BB^{N\ell}$ be the PRG constructed in 
\Cref{thm:prg} for $r=O(\log N\ell \cdot (\log\log N\ell)^3)=O(\log N \cdot 
(\log\log N)^3)$ and for $\epsilon = 1/N^{10c}$. For a string $s$ of length $r$ 
we define 
the hash function $h_s(i)$ as follows. First, it computes $y=\PRG(s)$. 
Then, it interprets $y$ as $N$ blocks where each block 
is of length $\ell$ bits, and outputs 1 if 
and only if all the bits of the $\ith{i}$ block are 1. 
Formally, we define
$
h_s(i) = \bigwedge_{j=(i-1)\ell+1}^{i\ell}\PRG(s)_{j}.
$
We show that properties (i) and (ii) hold for the set $Z_{h_s}$ where $h_s \in 
\cH$. 
We begin with property (i). For $i \in [N]$ let $X_i=h_s(i)$ be a random 
variable where $s \sim \BB^r$. Moreover, let $X=\sum_{i=1}^{N}X_i$. Using this 
notation we have that $|Z_{h_s}|=X$.
Thus, to show property (i), we need to show that
$
\Pr_{s \sim \BB^r}[X = O(N/\Delta)] \ge 2/3.
$
Let $f_i \colon \BB^{N\ell} \to \BB$ be a function that outputs 1 if the 
$\ith{i}$ 
block is all 1's. That is,
$
f_i(y)=\bigwedge_{j=(i-1)\ell+1}^{i\ell}y_{j}.
$
Since $f_i$ 
is a read-once DNF formula we have that
$$
\left |\Exp{y \sim \BB^{N\ell}}{f_i(y)} - \Exp{s \sim \BB^{r}}{f_i(\PRG(s))} 
\right| 
\le \epsilon.
$$
Therefore, it follows that
$$
\E{X} = \sum_{i=1}^{N}\E{X_i} = \sum_{i=1}^{N}\Exp{s \sim 
\BB^{r}}{f_i(\PRG(s))} \le \sum_{i=1}^{N}(\Exp{y \sim \BB^{N\ell}}{f_i(y)} + 
\epsilon) = N(2^{-\ell} + \epsilon) \leq 2c' N/\Delta.$$
Then, by Markov's inequality we get that $\Pr_{s \sim \BB^r}[X > 3\E{X}] 
\le 1/3$ and thus
$$
\prob{s \sim \BB^r}{X \le 6c'\cdot N/\Delta} \ge
1-\prob{s \sim \BB^r}{X > 3\E{X}} \ge 2/3.
$$

We turn to show property (ii). Let $S$ be any set of size at least $\Delta$ and 
let $g \colon \BB^{N\ell} \to 
\BB$ be an indicator function for the event that the set $S$ is covered. That 
is,
$$
g(y) = \bigvee_{i \in S}\bigwedge_{j=(i-1)\ell+1}^{i\ell}y_j.
$$
Since $g$ is a read-once DNF formula, we have that 
\begin{equation}\label{eq:hit}
\left | \prob{y \sim \BB^{N\ell}}{g(y)=0} - \prob{s \sim \BB^{r}}{g(\PRG(s))=0} 
\right| \le \epsilon~.
\end{equation}

Let $q \colon \BB^{N\ell} \to \{0,|S|\}$ be defined as follows: 
$$
q(y)=
\begin{cases}
0, \text{~~if~~} g(y)=0,\\
|S|, \text{~~otherwise.}\\
\end{cases}
$$
By definition, it holds that
$$\Exp{y \sim \BB^{N\ell}}{q(y)}=|S|\cdot \prob{y \sim \BB^{N\ell}}{g(y)=0} \mbox{~and~} \Exp{s \sim \BB^{r}}{q(\PRG(s))}=|S|\cdot \prob{s \sim \BB^{r}}{g(\PRG(s))=0}.$$ 
Combining this with Eq. (\ref{eq:hit}) yields that:
\begin{eqnarray}\label{eq:hit-better}
\left | \Exp{y \sim \BB^{N\ell}}{q(y)} - \Exp{s \sim \BB^{r}}{q(\PRG(s))}
\right| &=& |S|\left | \prob{y \sim \BB^{N\ell}}{g(y)=0}-\prob{s \sim \BB^{r}}{g(\PRG(s))=0} \right|
\\&\le& |S|\cdot \epsilon \leq 1/N^{c-1}~.\nonumber
\end{eqnarray}
Therefore to bound $\Exp{s \sim \BB^{r}}{q(\PRG(s))}$ it is sufficient to bound $\Exp{y \sim \BB^{N\ell}}{q(y)}$.
%The probability that $S$ is not hit by a random bit string $y \sim \BB^{n\ell}$ is $(1-p)^|S|$. 
Since all the bits in $y$ are completely independent we have that
\begin{eqnarray*}
\Exp{y \sim \BB^{N\ell}}{q(y)}&=&\prod_i(1-\prob{y_i \sim \BB^{\ell}}{f_i(y)})\cdot |S|
\\&= & (1-p)^{|S|}\cdot |S|\leq e^{-|S|/\Delta}\cdot |S|= O(\Delta)~,
\end{eqnarray*}
where $y_i$ is the $\ith{i}$ block of $\ell$ bits in $y$. 
Property (2) follows by combining with Eq. (\ref{eq:hit-better}).
%
%
%
%
%Let $Y_i= \bigwedge_{j=(i-1)\ell+1}^{i\ell}y_j$, and let $Y=\sum_{i \in S}Y_i$. 
%Then 
%$
%\E{Y} = \sum_{i \in S}\E{Y_i} \ge |S| 2^{-\ell} \ge c'|S|/\Delta.
%$
%
%For $y \in \BB^{n\ell}$, let $R(y)=\{ i ~\mid ~ h_s(i)=1\}$ be the set of elected nodes based on $y$.
%Let $q \colon \BB^{n\ell} \to [0, |S|]$ be a function that outputs 0 if $S \cap R(y)\neq \emptyset$ and $|S|$ otherwise. 
%$$
%\left |\Exp{y \sim \BB^{n\ell}}{f_i(y)} - \Exp{s \sim \BB^{r}}{f_i(\PRG(s))} 
%\right| 
%\le \epsilon.
%$$
%
%
%
%Thus, by a Chernoff bound we have that $\Pr[Y = 0 ] \le \Pr[\E{Y} - Y \ge c'\log n ]  \le 1/n^{2c}$, for a large enough constant $c'$ (that depends on $c$). Together, we get that \\
%$
%\Pr_s[ S \cap Z_{h_s} \ne \emptyset ] = \Exp{s \sim \BB^{r}}{g(\PRG(s))} 
%\ge  \Exp{y \sim \BB^{n\ell}}{g(y)} - \epsilon 
%= \prob{y \sim \BB^{n\ell}}{Y \ge 1} - \epsilon 
%%& \ge 1-1/n^{2c}-\epsilon 
%\ge 1-1/n^{c}.
%$
\end{proof}

We next present a deterministic construction of the \emph{soft} hitting set by means of derandomization. The round complexity of the algorithm depends on the number of random bits used by the randomized algorithms.
\begin{theorem}\label{thm:generalderand}
Let $G=(V,E)$ be an $n$-vertex graph, let $L,R \subset V$ be such that each $u \in L$ has a subset $S_u \subseteq R$ of size at least $\Delta$. Let $N=|R|$ and $c$ be a constant.
Let $\cH = \{h \colon [N] \to \BB \}$ be a family of hash functions such 
that choosing a random function $h$ from $\cH$ takes $g(N,\Delta)$ 
random bits and for $Z_h = \{u \in [N] : h(u)=0 \}$ it holds that: 
(1) $\pr{ |Z_h|=O(N/\Delta)} \ge 2/3$ and
(2) $\sum_{u \in L}\left(\pr{Z_h\cap S_u=\emptyset}\cdot |S_u|\right)=O(\Delta \cdot |L|)$. \\
Then, there exists a deterministic algorithm $\cA_{det}$ that constructs a soft hitting set in $O(g(N,\Delta)/\log N)$ rounds. 
\end{theorem}
\begin{proof}
Our goal is to completely derandomize the process of finding $Z_h$ by using the 
method of conditional expectations. We follow the scheme of 
\cite{Censor-HillelPS16} to achieve this, and define a unified cost function that depends on the two properties of the soft hitting set. %bad events that can occur when using a random seed of size $g=g(n,\Delta)$.  
Let $X=|Z_h|$ and let $Y=\sum_{u \in L}SH(S_u,Z_h) \cdot \chi$ where $\chi=N/(\Delta^2 \cdot |L|)$. The term $\chi$ is a normalization factor in order to make $X$ and $Y$ to have the same order of magnitude. This factor can be known to all nodes in $O(1)$ rounds of communication. 
The allows us to apply the method of conditional expectation to minimize $X+Y$.
Based on Lemma \ref{lem:hittingrandloglog}, $\Exp{h \in \cH}{X+Y}=O(N/\Delta)$. Note that the set $Z_h$ that minimizes the above expectation is indeed a soft-hitting set. 

Our goal is to find a seed of $g$ bits that minimizes $X+Y$ via the 
method of conditional expectations. In each step of the method, we run a distributed protocol to compute the conditional expectation based on a fixed prefix of the seed. 
Fix a random seed $y$ of length $g$, and a function $h$ chosen from $\cH$ using $y$. 
For every vertex $u \in R$, let $x_u=1$ if $u$ is in $Z_{h}$. In addition, for every vertex $v \in L$, let $y_v=|S_v| \cdot \chi$ if $S_v \cap Z_h=\emptyset$ and $y_v=0$ otherwise\footnote{Note that $L$ and $R$ might be overlapping.}. We then have that $\sum_{u \in R} x_u+ \sum_{v \in L} y_v= X+Y$.
 
%We denote by $X_v$ the indicator variable for the event that vertex $v$ adds too many edges.
%Letting $X_u$ be indicator random variable for the event that $S_u$ is not hit 
%by $Z_h$,
%we can write our expectation as follows:
%$
%\E{A}+\E{X_B} = \pr{X_A=1} + SH() = \pr{X_A=1} + \pr{\vee_{u} X_u=1}~$
%Suppose we have a partial assignment to the seed, denoted by $Y$. Our goal is to compute the conditional expectation $\E{X~|~Y}$, which translates to computing $\pr{X_A=1|Y}$ and $\pr{\vee_{u} X_u=1|Y}$. Notice that computing $\pr{X_A=1~|~Y}$ is simple since it depends only on $Y$ (and not on the graph or the subsets $\mathcal{S}$). The difficult part is computing $\pr{\vee_{u} X_u=1|Y}$. 
%
%Instead, we use a pessimistic estimator of $\E{X}$ which avoids this difficult computation. Specifically, we define the estimator:
%$\Psi = X_A + \sum_{u \in L} X_u.$
%Recall that for any $u \in L$ for a random $g$-bit length seed, it holds that $\pr{X_u=1}\leq 1/n^c$ and thus by applying a union bound over all $n$ sets, it also holds that
%$\E{\Psi} = \pr{X_A=1} + \sum_u \pr{X_u=1} < 1$.
%
%
%We describe how to compute the desired seed using the method of conditional expectation. We will reveal the assignment of the seed in chunks of $\ell=\lfloor \log n \rfloor$ bits. In particular, we show how to compute the assignment of $\ell$ bits in the seed in $O(1)$ rounds. Since the seed has $g$ many bits, this will yield an $O(g/\log n)$ round algorithm. 

Suppose that we are given a partial assignment to the first $k$ bits in the seed, and let $Q$ be a random variable that denotes the seed conditioned on this prefix. Any vertex can easily compute the conditional
expectation of $\Exp{}{x_u+y_u ~\mid~ Q}$. The reason is that the expectation depends on the probability that $u$ or one of the vertices in $S_u$ joins the soft hitting set. Since $u$ knows all the vertices in $S_u$ it can compute this value. 

Consider the $\ith{i}$ chunk of the seed $Q_i=(q_1,\ldots, q_\ell)$ for $\ell =\lfloor \log N \rfloor$, and assume that the assignment for the first $i-1$ chunks $Q_1\ldots,Q_{i-1}$ have been computed. For each of the $N$ possible assignments to $Q_i$, we assign a vertex $v$ that receives the conditional expectation values $\Exp{}{x_u+y_u~|~Q_1,\ldots,Q_{i}}$ from all vertices $u \in L$. 
The vertex $v$ then sums up all these values and sends them to a global leader $w$. The leader $w$, using the values it received from all the vertices, can compute $\Exp{}{X+Y~|~Q_1,\ldots, Q_i}$ for each of the possible $N\leq n$ assignments to $Q_i$. Finally, $w$ selects the assignment $(q^*_1,\ldots,q^*_\ell)$ that minimizes the expectation, and broadcasts it to all vertices in the graph. After $O(g/\log N)$ rounds $Q$ has been completely fixed such that the cost value of the chosen set $Z$ is at most $O(N/\Delta)$.
\end{proof}

\begin{proof}[Lemma \ref{lemma_soft_hitting}]
The proof follows by combining Theorem \ref{thm:generalderand} with Lemma \ref{lem:hittingrandloglog}.
Specifically, Theorem \ref{thm:generalderand} gives an $O(g(N,\Delta)/\log N)$-round algorithm for computing a soft hitting set. The term $g(N,\Delta)$ is the number of random bits required to select a random function from a family $\cH = \{h \colon [N] \to \BB \}$ which satisfies the two crucial properties stated in Theorem \ref{thm:generalderand}. Lemma \ref{lem:hittingrandloglog} shows the construction of such a family $\cH$ that requires only $g(N,\Delta)=O(\log N \cdot (\log\log N)^3)$ random bits in order to select a function $h$ from $\cH$ uniformly at random. Using this hash family in Theorem \ref{thm:generalderand} gives an $O(\log\log^3 n)$ algorithm. The lemma follows.
\end{proof}

\section{Discussion}

We showed here $poly(\log{\log n})$ algorithms for approximate shortest paths in unweighted graphs. While this improves exponentially on previous results, many interesting questions are still open.

First, our results here are for \emph{unweighted} graphs, and a natural question is whether it is possible to obtain similar results also for the weighted case.
Another question is whether the complexity of approximating shortest paths in the \clique can be improved even further. We remark that our approach seems to get stuck at $poly(\log \log n)$ for the following reason. If we want to build a $(1+\epsilon, \beta)$-emulator of near linear size, $\beta$ has to be at least logarithmic due to existential results \cite{abboud2018hierarchy}, which seems to lead to a complexity of at least $poly(\log \beta) = poly(\log \log n)$.

Third, we get $(1+\epsilon,\beta)$-approximation for APSP for $\beta = O(\frac{\log \log{n}}{\epsilon})^{\log{\log n}}.$ An interesting question is to obtain efficient algorithms with smaller $\beta$, for example $\beta = poly(\log n)$. We remark that current constructions of near linear size $(1+\epsilon, \beta)$-emulators, all have $\beta = O(\frac{\log \log{n}}{\epsilon})^{\log{\log n}}$, and it is an open question to get emulators with better parameters. However, it may be possible to get better approximations for shortest paths without constructing emulators.  

Finally, while we showed fast algorithms for approximate shortest paths in unweighted undirected graphs, currently the fastest algorithms for \emph{exact} shortest paths, even in unweighted undirected graphs, are still polynomial. Similarly, the fastest approximate algorithms for directed shortest paths are polynomial. Obtaining faster algorithms is an interesting open question.

\paragraph{Acknowledgments.} We would like to thank Keren Censor-Hillel for many fruitful discussions. The research is supported in part by the Israel Science Foundation (grant no. 1696/14 and 2084/18), the European Union's Horizon 2020 Research and Innovation Program under grant agreement no. 755839, and the Minerva grant no. 713238.

\bibliographystyle{plainurl}
\bibliography{near_additive}

\appendix

\section{Comparison to existing nearly-additive emulator constructions} \label{sec:comparison}

Here we discuss in more detail the connections between our construction and the centralized emulator constructions of Elkin and Neiman \cite{elkin2018efficient} and Thorup and Zwick \cite{thorup2006spanners}.
The algorithm of Elkin and Neiman \cite{elkin2018efficient} is based on the superclustering and interconnection approach that was introduced in \cite{elkin2004} and appears also in \cite{elkin2018efficient, elkin2019near, elkin2019hopsets}, in constructions of spanners, emulators and hopsets. 
At a very high-level the algorithm maintains clusters, where in each iteration clusters that have many neighbouring clusters are merged to super-clusters (superclustering), and unclustered clusters add edges between them (interconnection). While our algorithm does not build any clusters explicitly, there are similarities between the processes, in terms of the sampling rates of the
hierarchical clustering procedure, and the radii of these clusters. The key difference is that EN takes a cluster-centric perspective whereas our algorithm takes a vertex-centric perspective. This might appear somewhat semantic, however, it appears to be useful for distributed implementation. We note that the description of our algorithm is very short and, in a sense, it is also simpler. 

The TZ algorithm seems less related, at least at a first glance, to our construction due to its global nature. More specifically, in the TZ algorithm, every vertex is required to collect information from its $n$-radius ball, while in our algorithm  the exploration radius is $\beta = O(\frac{\log{\log{n}}}{\epsilon})^{\log{\log{n}}}$. A deeper look suggests, however, that our algorithm can be viewed as localized variant of the TZ algorithm. This view might provide an alternative explanation for the universality of TZ algorithm \cite{EN-HopsetSurvey20}. In more detail, the description of the algorithm of TZ seems much closer to our algorithm: it is based on sampling sets of vertices $\emptyset = S_r \subset S_{r-1} \ldots \subset \ldots S_1 \subset S_0=V$, and then adding edges in the following way: each vertex in $S_i$ adds edges to the closest vertex from $S_{i+1}$ (if exists), and to all vertices in $S_i$ that are strictly closer. Note that in our algorithm each vertex in $S_i$ has some exploration radius $\delta_i$ that depends on $\epsilon$ and it adds edges to the closest vertex from $S_{i+1}$ at distance $\delta_i$ if exists, and otherwise it adds edges to all vertices in $S_i$ at distance at most $\delta_i$. If we look at these two descriptions, we can see that all the edges taken to our emulator, for \emph{any} choice of $\epsilon$, are contained in the emulator built by TZ. This explains the fact that   
the TZ emulator is \emph{universal} in the sense that the same $(1+\epsilon,\beta)$-emulator works for all $\epsilon$.  

Finally, we note that our emulator algorithm was designed with the purpose of being efficiently implemented in the \clique model. We did not optimize for the sparsity of the emulator and its additive term $\beta$, indeed, these are somewhat tighter in the constructions of \cite{elkin2018efficient,thorup2006spanners}. 
For a more in-depth survey for current emulator constructions, see the recent survey of Elkin and Neiman \cite{EN-HopsetSurvey20}.

\section{Toolkit}

\subsection{Hitting sets} \label{sec:proof_hitting}

We next provide a proof for Lemma \ref{rand_hit}.

\rhitting*

\begin{proof}
To construct the hitting set $A$, we add each vertex to $A$ independently with probability $\frac{c \ln{n}}{k}$ for a constant $c > 2$. The expected size of the set is clearly $c n\ln{n}/k$. From a Chernoff bound, we get that $Pr[A > (1+\delta) c n\ln{n}/k] \leq e^{-{\frac{\delta^2{c n\ln{n}/k}}{2+\delta}}}$ for any $\delta > 0$. Choosing $\delta=2$ and using the fact that $k \leq n$ shows that $|A| \leq 3 c n\ln{n}/k = O(n\log{n}/k)$ w.h.p. 

We next show that $A$ is a hitting set w.h.p. As $A$ has each vertex with probability $\frac{c \ln{n}}{k}$, and the sets $S_v$ are of size at least $k$, the probability that $S_v \cap A = \emptyset$ is $(1-\frac{c \ln{n}}{k})^k \leq e^{-c \ln{n}} \leq \frac{1}{n^{c}}$. Using union bound, we get that w.h.p for all vertices $v \in V'$, there is a vertex in $S_v \cap A.$
\end{proof}

\subsection{The \kdnearest} \label{sec:nearest}

In this section, we prove Theorem \ref{thrm:knearest}, the proof basically follows \cite{DBLP:conf/podc/Censor-HillelDK19} and we mainly focus on the differences. 

\nearest*

In \cite{DBLP:conf/podc/Censor-HillelDK19} a similar theorem is proven and holds also for the weighted case, but the log factors are $\log{n}$ and $\log{k}$. Since our goal is to get a sub-logarithmic complexity, we show that these log factors can be made $\log{d}$ if we only focus on vertices at distance at most $d$ in unweighted graphs. The proof requires an algorithm for \emph{filtered} matrix multiplication described in \cite{DBLP:conf/podc/Censor-HillelDK19}. We next review the main ingredients needed, for full details see \cite{DBLP:conf/podc/Censor-HillelDK19}.

\paragraph{Distance products.} We would be interested in multiplying matrices in the min-plus semiring. In this setting, $(S \cdot T) [i,j] = \min_k{(S[i,k]+T[k,j])}$. This corresponds to computing distances in the following sense. If $A$ is the adjacency matrix of the graph, the matrix $A^2[u,v]$ contains the shortest path between $u$ and $v$ that uses at most two edges, similarly $A^i[u,v]$ contains the shortest path between $u$ and $v$ that uses at most $i$ edges. The zero element in this semiring is $\infty$ which corresponds to having no edge between vertices.  
Since we are interested in learning only $k$ closest vertices to each vertex, we would be interested in a filtered version for matrix multiplication we define next.

\paragraph{Filtered matrix multiplication.} Given a parameter $\rho$, and a matrix $P$, we denote by $\overline{P}$ a matrix that in each row $i$ contains only the smallest $\rho$ non-zero entries in row $i$ of $P$, other entries contain zero. In the filtered matrix multiplication problem, we are given two matrices $S,T$ and our goal is to get $\overline{P}$, where $P = S \cdot T$. I.e., we are interested in a filtered version of the product. In \cite{DBLP:conf/podc/Censor-HillelDK19} it is shown that filtered matrix multiplication can be done in a complexity that depends on the density of $S$ and $T$ and the value $\rho$. Here, the density $\rho_S$ of a matrix $S$ is the average number of non-zero entries in a row of $S$. The following is proven.

\begin{theorem}\label{theorem:mm-filtered}
Filtered matrix multiplication can be computed in
\[O\biggl( \frac{( \rho_S \rho_T \rho )^{1/3}}{n^{2/3}} + \log W \biggr)\]
rounds in the \clique, where $W$ is a bound on the number of possible values for the semiring
elements, and $W$ is known by the vertices.
\end{theorem}

Note that in the context of distance products in unweighted graphs, if we only look at close by vertices at distance at most $d$, there are only $O(d)$ possible values for the distances, hence $\log W = O(\log{d}).$ 

To compute the \kdnearest we would like to use filtered matrix multiplication $\log{d}$ times. Intuitively, after $i$ multiplications we would get the distances to the $k$-nearest vertices of distance at most $2^i$, hence after $\log{d}$ iterations we solve the \kdnearest problem.

We follow the following iterative algorithm. Let $A_1$ be the adjacency matrix of the graph, recall that $\overline{A_1}$ is a matrix obtained from $A_1$ be keeping only the $\rho$ smallest non-$\infty$ entries in each row. Let $A_2 = \overline{A_1} \cdot \overline{A_1}$. We use Theorem \ref{theorem:mm-filtered} to compute $\overline{A_2}.$ We continue iteratively, where in iteration $i$ we compute $\overline{A_{i+1}}$ where $A_{i+1} = \overline{A_i} \cdot \overline{A_i}.$ We continue for $\log{d}$ iterations.

\begin{claim}
Let $u,v$ such that $d(u,v) \leq 2^{i-1}$, then the matrix $\overline{A_i}[u,v]$ either contains $d(u,v)$ or $\infty$. In the latter case, the row of $u$ contains distances to $k$ nearest vertices, all of them at distance at most $d(u,v).$ 
\end{claim}

\begin{proof}
The proof is by induction. For $i=1$, it follows from the definition of $\overline{A_1}$. Assume it holds for $i$ and we prove it for $i+1$.
Let $u,v$ such that $d(u,v) \leq 2^i$, then there is a vertex $w$ in the shortest $u-v$ path such that $d(u,w) \leq 2^{i-1}$ and $d(w,v) \leq 2^{i-1}$. From the induction hypothesis, $\overline{A_i}[u,w]$ either contains $d(u,w)$ or $\infty$, and the same holds for $\overline{A_i}[w,v]$ with respect to $d(w,v).$

\emph{Case 1: $\overline{A_i}[u,w] = d(u,v)$ and $\overline{A_i}[w,v] = d(w,v)$}. In this case, $(\overline{A_i} \cdot \overline{A_i})[u,v]=d(u,v)$ by definition of min-plus multiplication and since $w$ is on the shortest $u-v$ path. After filtering, the matrix $\overline{A_{i+1}}[u,v]$ either contains $d(u,v)$ or $\infty$, but the latter can happen only if the row of $u$ contains $k$ nearest vertices, all of them at distance at most $d(u,v).$

\emph{Case 2: $\overline{A_i}[u,w] = \infty$.} In this case, by the induction assumption, the row of $u$ in $\overline{A_i}$ contains $k$-nearest vertices of distance at most $d(u,w)$ from $u$. All these entries are kept in the matrix $A_{i+1} = (\overline{A_i} \cdot \overline{A_i})$. After filtering, it follows that $\overline{A_{i+1}}$ contains $k$-nearest vertices to $u$ of distance at most $d(u,w) \leq d(u,v).$

\emph{Case 3: $\overline{A_i}[u,w] = d(u,v)$ and $\overline{A_i}[w,v] = \infty$.} In this case, by the induction assumption, the row of $w$ in $\overline{A_i}$ contains $k$-nearest vertices of distance at most $d(w,v)$ from $w$. Denote these vertices by $v_1,...,v_k$. Since we can go from $u$ to $w$ and then to $v_i$, it holds that $$A_{i+1}[u,v_i] = (\overline{A_i} \cdot \overline{A_i})[u,v_i] \leq d(u,w)+d(w,v_i) \leq d(u,w)+d(w,v)=d(u,v).$$ Hence, $A_{i+1}$ includes distances to at least $k$ vertices of distance at most $d(u,v)$ from $u$. After filtering, we have that $\overline{A_{i+1}}$ contains $k$-nearest vertices to $u$ of distance at most $d(u,v).$ 
\end{proof}

Hence, after $\log{d}$ iterations, we have the matrix $\overline{A_{\log{d}+1}}$ that contains distances to the $k$-nearest vertices to each vertex $v$ of distance at most $2^{\log{d}}=d.$ Note that all the distances computed in the process are at most $d$, as we start with the adjacency matrix that has all distances at most $1$, and in each iteration the distances can at most duplicate, hence the number of elements in the semiring is $O(d)$. We use filtered matrix multiplication $\log{d}$ times with all 3 matrices of density $\rho$, as we keep only filtered version of the matrices. Hence, by Theorem \ref{theorem:mm-filtered} we get a complexity of $$O\biggl( \biggl( \frac{\rho}{n^{2/3}} + \log d \biggr) \log d \biggr).$$ This concludes the proof of Theorem \ref{thrm:knearest}. 

\subsection{Bounded hopsets} \label{sec:hopsets}

We next explain how to build $(\beta,\epsilon,t)$-hopsets with $\beta=O(\frac{\log{t}}{\epsilon})$ in $(\frac{\log^2{t}}{\epsilon})$ time. The construction mostly follows the construction in \cite{DBLP:conf/podc/Censor-HillelDK19} for $(\beta,\epsilon,n)$-hopsets with $\beta=O(\frac{\log{n}}{\epsilon})$ in $(\frac{\log^2{n}}{\epsilon})$ time, with some changes to exploit the fact that we consider only paths of at most $t$ edges. 
We next discuss the construction in \cite{DBLP:conf/podc/Censor-HillelDK19}, and the changes in the algorithm and analysis required in our case.

\subsubsection*{Overview of the construction from \cite{DBLP:conf/podc/Censor-HillelDK19}} The construction is based on recent hopsets constructions \cite{DBLP:journals/corr/ElkinN17, huang2019thorup} that are based on the emulators of Thorup and Zwick \cite{thorup2006spanners}. In \cite{DBLP:conf/podc/Censor-HillelDK19}, it is shown how to implement a variant of these constructions in just poly-logarithmic number of rounds in the \clique. We now describe this variant.

Given a graph $G=(V,E)$, we define the sets $V=A_0 \supseteq A_1 \supseteq A_2 = \emptyset$, where $A_1$ is a hitting set of size $O(\sqrt{n})$ such that each vertex has at least one vertex from $A_1$ among the $\sqrt{n} \log{n}$ closest vertices.
For a given subset $A \subseteq V$, we denote by $d_G(v,A)$ the distance from $v$ to the closest vertex in $A$.
For a vertex $v \in V$, we denote by $p(v) \in A_1$ a vertex of distance $d_G(v,A_1)$ from $v$. 
For a vertex $v \in A_0 \setminus A_1$, we define the \emph{bunch}
$$B(v) = \{u \in A_0: d_G(v,u) < d_G(v,A_1) \} \cup p(v),$$
and for a vertex $v \in A_1$, we define the bunch $B(v) = A_1$.

The hopset is the set of edges $H=\{\{v,u\}: v \in V, u \in B(v)\}$. The length of an edge $\{u,v\}$ would include $d_G(u,v)$ or an approximation to it computed throughout the algorithm. 

To construct the hopsets efficiently in the \clique, first bunches of vertices in $A_0 \setminus A_1$ are computed using an algorithm for computing the $k$-nearest vertices. Then, an iterative process is used to compute the distances to the bunches of vertices in $A_1$, where in iteration $i$, a $(\beta,\epsilon i, 2^i)$-hopset is constructed. The whole construction requires poly-logarithmic time.

\subsubsection*{Faster construction for bounded hopsets}

We next discuss the changes needed to get a faster algorithm for $(\beta,\epsilon,t)$ hopsets. We focus on constructing hopsets only for \emph{unweighted} graphs.
The construction in \cite{DBLP:conf/podc/Censor-HillelDK19} requires poly-logarithmic number of rounds for two reasons. First, for computing the $k$-nearest vertices (in the entire graph), and second for the $\log{n}$ iterations required to build a $(\beta,\epsilon,n)$-hopset. Since in our case we only need a $(\beta,\epsilon,t)$-hopset, $\log{t}$ iterations would be enough for the second part. For the first part, we would like to replace computing the $k$-nearest by computing the \ktnearest, i.e., the $k$-nearest of distance at most $t$. This however, allows to compute only bounded part of the bunches of vertices, and we next discuss the changes required in the algorithm and analysis when we work only with bounded bunches.

We define the \emph{bounded bunch} of a vertex $v$ to be $B_t(v) = B(v) \cap B(v,t,G)$, i.e., it includes only the part of the bunch of distance at most $t$ from $v$. The hopset is the set of edges $$H'=\{\{v,u\}: v \in V, u \in B_t(v)\}.$$  The weights of the edges again would include distances or approximations for them computed in the algorithm. We next explain how to implement the algorithm efficiently, and prove its correctness.

\paragraph{Fast implementation.} 

First, we need to construct the hitting set $A_1$, we show two variants, randomized and deterministic. We denote by $N_{k,t}(v)$ the set of $k = \sqrt{n} \log{n}$ closest vertices to $v$ of distance at most $t$ computed by the \ktnearest algorithm. The set $A_1$ would be a hitting set for the sets $N_{k,t}(v)$ for all vertices where $|N_{k,t}(v)| = k$, i.e., vertices that have at least $k$ vertices at distance $t$. While in \cite{DBLP:conf/podc/Censor-HillelDK19}, the hitting set $A_1$ was a hitting set of sets $N_k(v)$ of the $k$ nearest vertices to $v$ in the whole graph, in our case, since we look at bounded bunches, it is enough to compute a hitting set for the sets $N_{k,t}$ as described above.
We can either construct $A_1$ randomly using Lemma \ref{rand_hit} or deterministically using Lemma \ref{det_hit}, which give the following. 

\begin{claim} \label{hitting_claim}
We can construct a hitting set $A_1$ of all the sets $N_{k,t}(v)$ for vertices where $|N_{k,t}(v)|=k$, in the following ways:
\begin{enumerate}
\item In one round, where at the end $A_1$ is a hitting set of size $O(\sqrt{n})$ w.h.p.
\item In $O(\log^2{t} + (\log{\log{n}})^3)$ rounds, where at the end $A_1$ is a hitting set of size $O(\sqrt{n})$.
\end{enumerate}
At the end of the computation, all vertices know which vertices are in $A_1.$
\end{claim}

\begin{proof}
We can either construct $A_1$ randomly using Lemma \ref{rand_hit}, which requires only one round of communication to let all vertices learn which vertices are in $A_1$. Alternatively, we can construct $A_1$ deterministically using Lemma \ref{det_hit}. Then, we first need that each vertex $v$ would learn the set $N_{k,t}(v)$, which takes $O(\log^2{t})$ rounds using the \kdnearest algorithm. Afterwards, the time complexity for constructing the hitting set is $O((\log{\log{n}})^3)$ rounds. In both cases we have that the size of $A_1$ is $O(\sqrt{n}).$
\end{proof}

For the rest of the section, we assume that $A_1$ is indeed a hitting set of size $O(\sqrt{n})$, which happens w.h.p if we build $A_1$ using a randomized algorithm, or holds always if we construct $A_1$ using a deterministic algorithm. The rest of the claims hold w.h.p in the randomized case, and hold always in the deterministic case. We first bound the size of the hopset.

\begin{claim} \label{edges_claim}
The number of the edges in the hopset is $O(n^{3/2} \log{n})$.
\end{claim}

\begin{proof}
For a vertex $v \in A_0 \setminus A_1$, we have $B_t(v) = B(v) \cap B(v,t,G)$. If $|B(v,t,G)| \leq k$, then clearly the bunch of $v$ is of size at most $k$, hence $v$ adds at most $k = \sqrt{n} \log{n}$ edges to the hopset. Otherwise, we have that $|N_{k,t}(v)| = k$. Since $A_1$ is a hitting set of the sets $N_{k,t}(v)$ for vertices where $|N_{k,t}(v)| = k$, we have that there is a vertex $u \in A_1$ among the $k$ closest vertices to $v$. Since $v$ adds to the hopset only edges to $u$ and vertices that are strictly closer than $u$, it adds at most $k$ edges to the hopset as needed.
Overall, all vertices in $A_0 \setminus A_1$ add $O(n k) = O(n^{3/2} \log{n})$ edges to the hopset.
Vertices from $A_1$ only add edges to vertices in $A_1$, since $A_1$ is of size $O(\sqrt{n})$ this adds $O(n)$ edges. 
\end{proof}

We next explain how to compute the bounded bunches.
To compute bounded bunches of vertices in $A_0 \setminus A_1$ we use the \kdnearest algorithm, which gives the following.

\begin{claim} \label{bunches0}
We can compute for all $v \in A_0 \setminus A_1$, the distances to the vertices in $B_t(v)$ in $O(\log^2{t})$ rounds.
\end{claim}

\begin{proof}
We use the \kdnearest algorithm with $k=\sqrt{n} \log{n}, d=t$, we now show that this allows computing $B_t(v)$ for all vertices in $A_0 \setminus A_1$. Let $v \in A_0 \setminus A_1$, if $|B(v,t,G)| \leq \sqrt{n} \log{n}$, then by computing the \ktnearest, $v$ actually computed the distances to the whole set $B(v,t,G)$. This allows computing distances to $B_t(v)$ which includes the closest vertex in $A_1 \cap B(v,t,G)$ if such exist, and all vertices strictly closer than it, or the whole set $B(v,t,G)$ if it does not include any vertex from $A_1$. Otherwise, $|B(v,t,G)| > \sqrt{n} \log{n}$. In this case, the $k$-closest vertices to $v$ are at distance at most $t$, which gives $|N_{k,t}(v)| = k$.
By the definition of $A_1$, there is a vertex from $A_1$ in the \ktnearest vertices to $v$. Hence, after computing the distances to the \ktnearest, $v$ knows the distance to the closest vertex from $A_1$, and also to all vertices that are strictly closer than it. Hence, it knows the distances to all vertices in $B_t(v).$
\end{proof}

To compute the bounded bunches of vertices in $A_1$, we follow the following iterative process described in \cite{DBLP:conf/podc/Censor-HillelDK19}. In iteration $1 \leq \ell \leq \log{t}$ we build a $(\beta,\epsilon \ell, 2^{\ell})$-hopset $H^{\ell}$, as follows. 
Let $H^0 = \{ \{u,v\} | u \in A_0 \setminus A_1, v \in B_t(v) \}$, we already computed all these edges and their weights when computing bounded bunches of vertices in $A_0 \setminus A_1$, and include them as part of $H^{\ell}$ for all $\ell$. Note that any set (including the emptyset) is a $(\beta, 0,1)$-hopset for any $\beta \geq 1$. Hence, $H^0$ is a $(\beta,0,1)$ hopset.

To construct $H^{\ell}$ for $1 \leq \ell \leq \log{n}$ we work as follows. Let $G' = G \cup H^{\ell -1}$.
All vertices in $A_1$ compute distances to all vertices in $A_1$ at hop-distance at most $4\beta$ in the graph $G'$.
Each vertex $v \in A_1$, adds to the hopset $H^{\ell}$ all the edges $\{v,u\}$ where $u \in A_1$ at hop-distance at most $4\beta$ from $v$ in $G'$, with the weight it learned in the algorithm. The hopset $H^{\ell}$ includes all these edges, as well as all the edges of $H^0$. 

We next analyze the complexity of this process and then prove correctness.

\begin{claim} \label{bunches1}
Constructing $H^{\log{t}}$ requires $O(\beta \log{t})$ rounds.
\end{claim}

\begin{proof}
We build the sets $H^{\ell}$ in iterations, where we compute $H^{\log{t}}$ in iteration $\log{t}$.
we next show that each iteration takes $O(\beta)$ time. In iteration $\ell$, we would like to know distances to the sources from $A_1$ of distance at most $4\beta$ in the graph $G \cup H^{\ell-1}$, this can be implemented using the \sdk algorithm with $S=A_1, d=4\beta$. Since $|S|=O(\sqrt{n})$, by Theorem \ref{thrm:source-detection} the complexity is $O \biggl(\biggl( \frac{ n^{2/3} (\sqrt{n})^{2/3}}{n} +1 \biggr) \beta \biggl) = O(\beta)$ rounds.
\end{proof}

\paragraph{Correctness.} The correctness mostly follows \cite{DBLP:conf/podc/Censor-HillelDK19}, we focus on the differences. First, we need the following claim regarding the edges of $H^0$, a similar claim is proved in \cite{DBLP:conf/podc/Censor-HillelDK19} without restriction on the distance $d(x,y)$. 

\begin{claim} \label{basis_claim}
For any $x,y \in V$ with $d_G(x,y) \leq t$, either $d_{G \cup H^0}^1(x,y) = d_G(x,y)$ or there exists $z \in A_1$ such that $d_{G \cup H^0}^1(x,z) \leq d_G(x,y).$
\end{claim}

\begin{proof}
If $x \in A_1$, this follows trivially by setting $z=x$. We next focus on the case that $x \in A_0 \setminus A_1$. If $y \in B_t(x)$, then we added the edge $\{x,y\}$ with weight $d_G(x,y)$ to $H^0$. Otherwise, there is a vertex $p(x) \in B_t(x) \cap A_1$ where $d_G(x,p(x)) \leq d_G(x,y)$, and we added the edge $\{x,p(x)\}$ with weight $d_G(x,p(x))$ to $H^0$.
\end{proof}

The next lemma shows that $H^{\ell}$ is a $(\beta,\epsilon_{\ell},2^{\ell})$-hopset. A similar lemma appears in \cite{DBLP:conf/podc/Censor-HillelDK19} with the following changes: first, there $\epsilon < 1/\log{n}$ which is needed since the algorithm has $\log{n}$ iterations, in our case we can replace the $n$ with $t$. Second, the proof there considers hop-distance $2^{\ell}$ for any $\ell \leq \log{n}$, where in our case this is restricted to $\ell \leq \log{t}.$

\begin{lemma} \label{hopset_lemma}
Let $0< \epsilon < 1/\log{t}$. Set $\delta = \epsilon / 4$, and $\beta=3/\delta$ and let $x,y \in V$ be such that $d_G(x,y)=d_G^{2^{\ell}}(x,y)$ for $\ell \leq \log{t}$. Then, $$d_{G \cup H^{\ell}}^{\beta}(x,y) \leq (1+\epsilon_{\ell})d_G(x,y),$$ where $\epsilon_{\ell}=\epsilon \cdot \ell$.
\end{lemma}

The proof of the lemma is identical to \cite{DBLP:conf/podc/Censor-HillelDK19}. 
The only observation needed is the following. The proof uses a variant of Claim \ref{basis_claim} without a restriction on $d_G(x,y)$. In our case, in the lemma we only consider vertices where $d_G(x,y) \leq t$, for this case, the proof only uses the claim for vertices of distance at most $t$, hence Claim \ref{basis_claim} suffices.

\paragraph{Conclusion.}

We can now show Theorem \ref{hopset_thm}.

\hopset*

\begin{proof}
By Lemma \ref{hopset_lemma}, $H^{\log{t}}$ is a $(\beta,\epsilon {\log{t}},t)$-hopset. From the choice of parameters $0<\epsilon< 1/\log{t}$ and $\beta=O(1/\epsilon)$. Let $\epsilon' = \epsilon \log{t}$. Then, $\epsilon = \frac{\epsilon'}{\log{t}}$. Hence, $H^{\log{t}}$ is a $(\beta,\epsilon',t)$-hopset where $0< \epsilon' < 1$ and $\beta=O(\frac{\log{t}}{\epsilon'})$. 

By Lemma \ref{hitting_claim}, we can construct $A_1$ in one round using a randomized algorithm, or in $O(\log^2{t} + (\log{\log{n}})^3)$ rounds using a deterministic algorithm. Computing the bunches takes $O(\beta \log{t} + \log^2{t})$ rounds by Claims \ref{bunches0} and \ref{bunches1}. Since $\beta=O(\frac{\log{t}}{\epsilon'})$, the total time complexity is $O(\frac{\log^2{t}}{\epsilon'})$ rounds in the randomized case, and $O(\frac{\log^2{t}}{\epsilon'} + (\log{\log{n}})^3)$ rounds in the deterministic case. The number of edges is $O(n^{3/2} \log{n})$ by Claim \ref{edges_claim}.
\end{proof}

\section{Missing Proofs for the $(1+\epsilon, \beta)$-Emulator} \label{sec:emulator_app}

\subsection{Size analysis} \label{sec:emulator_app_size}

\sizeSi*

\begin{proof}
The proof is by induction. For $i=0$, it holds that $|S_0| = |V| = n = n^{1-\frac{2^0-1}{2^r}}.$
Assume it holds for $i$, and we will prove it for $i+1$. By definition, $S_{i+1} \gets Sample(S_i, n^{-\frac{2^{i}}{2^r}})$. Hence, the expected size of $S_{i+1}$ is  $$n^{1-\frac{2^i-1}{2^r}} \cdot n^{-\frac{2^i}{2^r}} = n^{1-\frac{2^i-1+2^i}{2^r}}= n^{1-\frac{2^{i+1}-1}{2^r}}.$$
\end{proof}

\probSr*

\begin{proof}
By the choice of parameters, a vertex is in $S_r$ with probability $$\prod_{i=1}^{r} p_r = p_r \cdot \prod_{i=1}^{r-1} n^{-\frac{2^{i-1}}{2^r}} = p_r \cdot \frac{1}{n^{\sum_{i=1}^{r-1} \frac{2^{i-1}}{2^r}}} = \frac{p_r}{n^{\frac{2^{r-1}-1}{2^r}}}.$$
Since $p_r = n^{-\frac{1}{2^r}}$, this probability is equal to
$$\frac{n^{-\frac{1}{2^r}}}{n^{\frac{2^{r-1}-1}{2^r}}} = \frac{1}{ n^{\frac{2^{r-1}}{2^r}}}= \frac{1}{\sqrt{n}}.$$
\end{proof}

\sizeSr*

\begin{proof}
By Claim \ref{claim_prob_Sr}, each vertex is in $S_r$ with probability $\frac{1}{\sqrt{n}}$. Also, by definition of the sampling process, each vertex is sampled independently of other vertices. The expected size of $S_r$ is clearly $\sqrt{n}.$ From a Chernoff bound, we get that $Pr[S_r > (1+\delta) \sqrt{n}] \leq e^{-{\frac{\delta^2{\sqrt{n}}}{2+\delta}}}$ for any $\delta > 0$. Choosing $\delta=1$ shows that $|S_r| \leq 2\sqrt{n}$ w.h.p. 
\end{proof}

\subsection{Stretch analysis} \label{sec:emulator_app_stretch}

\stretchRi*

\begin{proof}
The proof is by induction, for $i=0$ it holds trivially. Assume it holds for $i-1$, and we prove it holds for $i$.
$$R_i = \sum_{j=0}^{i-1} \delta_i = \sum_{j=0}^{i-1} (\frac{1}{\epsilon^j} + 2R_j)=\sum_{j=0}^{i-1} (\frac{1}{\epsilon^j} + 2\sum_{\ell = 0}^{j-1} \frac{1}{\epsilon^{\ell}} \cdot 3^{j-1 - \ell}).$$
In the sum $\sum_{\ell = 0}^{j'-1} \frac{1}{\epsilon^{\ell}} \cdot 3^{j'-1 - \ell}$, we will get the term $\frac{1}{\epsilon^j} \cdot 3^{j' - 1 - j}$ when $\ell = j$, this can only happen when $j \leq j' -1$. This gives, 
$$R_i = \sum_{j=0}^{i-1} \frac{1}{\epsilon^j} (1+2 \sum_{j' = j+1}^{i-1} 3^{j' - j - 1}).$$
Now $$\sum_{j' = j+1}^{i-1} 3^{j' - j - 1} = \frac{1}{3^{j+1}} \sum_{j'=j+1}^{i-1} 3^{j'} = \frac{1}{3^{j+1}} \cdot \frac{3^{j+1}(3^{i-j-1}-1)}{2} = \frac{3^{i-j-1}-1}{2}.$$
This gives, $$R_i = \sum_{j=0}^{i-1} \frac{1}{\epsilon^j} (1+2 \cdot \frac{3^{i-j-1}-1}{2}) = \sum_{j=0}^{i-1} \frac{1}{\epsilon^j} \cdot 3^{i-j-1}.$$
\end{proof}

\Ribound*

\begin{proof}
$$R_i = \sum_{j=0}^{i-1} \frac{3^{i-1-j}}{\epsilon^j} = 3^{i-1} \sum_{j=0}^{i-1} \frac{1}{(3\epsilon)^j} = 3^{i-1} \cdot \frac{\frac{1}{(3\epsilon)^i} -1}{\frac{1}{3\epsilon}-1} \leq 3^{i-1} \cdot \frac{\frac{1}{(3\epsilon)^i}}{\frac{1}{3\epsilon}(1-3\epsilon)} \leq \frac{3^{i-1}}{(3\epsilon)^{i-1}(1-3\epsilon)} \leq \frac{2}{\epsilon^{i-1}}.$$
\end{proof}

\Stretchtwo*

\begin{proof}
$$4R_i + 2 \beta_{i-1} = 4R_i +2(4\sum_{j=1}^{i-1} 2^{i-1-j}R_j) = 4R_i + 4 \sum_{j=1}^{i-1} 2^{i-j} R_j = 4\sum_{j=1}^{i} 2^{i-j} R_j = \beta_i.$$
\end{proof}

\Stretchone*

\begin{proof}
For $i=0$, this clearly holds. Otherwise, $\beta_i = 4 \sum_{j=1}^{i} 2^{i-j} R_j.$ From Claim \ref{claim_Ri_bound}, we have $R_j \leq \frac{2}{\epsilon^{j-1}}$. This gives, $$\beta_i \leq 4 \sum_{j=1}^{i} 2^{i-j} \frac{2}{\epsilon^{j-1}} = 4 \cdot 2^i \sum_{j=1}^{i} \frac{1}{(2\epsilon)^{j-1}}.$$
In addition,
$\sum_{j=1}^{i} \frac{1}{(2\epsilon)^{j-1}} = \sum_{j=0}^{i-1} \frac{1}{(2\epsilon)^j} = \frac{\frac{1}{(2\epsilon)^i}-1}{\frac{1}{2\epsilon}-1} \leq \frac{\frac{1}{(2\epsilon)^i}}{\frac{1}{2\epsilon}(1-2\epsilon)}=\frac{1}{(2\epsilon)^{i-1}(1-2\epsilon)}.$ This gives,
$$\beta_i \leq \frac{4 \cdot 2^i}{(2\epsilon)^{i-1}(1-2\epsilon)} \leq \frac{8}{\epsilon^{i-1}(1-2\epsilon)} \leq \frac{10}{\epsilon^{i-1}}.$$
\end{proof}

\subsection{Implementation in the \clique} \label{sec:emulator_app_clique}

\ConclusionExp*

\begin{proof}
The proof mostly follows the stretch analysis and conclusion in previous sections, here we only focus on the differences since we only approximated some of the distances. Note that for all edges with one endpoint outside $S_r$, we computed the correct distances. In particular, Claim \ref{claim_Ri} still holds, since it is only based on edges added by $i$-dense vertices in $S_i \setminus S_{i+1}$, which only happens in iterations $i < r$. Also, Lemma \ref{lemma_stretch} still holds for $i < r$, the only difference is for $i=r$. Since we computed $(1+\epsilon')$-approximations to the distances between vertices in $S_r$, it changes slightly the constants in the analysis for the case $i=r$. Now, for $r$-clustered vertices $u',v'$, we have that $d_H(c_r(u'),c_r(v')) \leq (1+\epsilon')(d(u',v')+2R_r)$ (previously we had $d_H(c_r(u'),c_r(v')) \leq d(u',v')+2R_r$).   
This gives $d_H(u',v') \leq (1+\epsilon')d(u',v')+(4+2\epsilon')R_r$. 
If we choose $\epsilon' =20\epsilon(r-1)$ (assuming $r>1$), we get that only constants in the analysis are changed. Specifically, in Case 1 of the proof of Lemma \ref{lemma_stretch} for the case $i=r$, we get that $$d_H(u,v) \leq (1+20\epsilon(r-1))d(u,v)+4R_r+2\beta_{r-1} + 2 \epsilon' R_r.$$ The only difference from before is the additional $2 \epsilon' R_r$ term in the stretch. As $\beta_i = 4R_i + 2\beta_{i-1}$ by Claim \ref{claim_stretch_2} and $\epsilon' < 1$, we have that $2 \epsilon' R_r < \beta_r$, which bounds the stretch in Case 1 with $(1+20\epsilon(r-1), 2 \beta_r).$ Case 2 now follows the original proof, with the change that we replace $\beta_r$ with $2 \beta_r$. This gives a bound of $(1+40\epsilon r, 2 \beta_r)$ on the stretch in Case 2.

Now we can follow the proof of Theorem \ref{thm_conclusion_emulator} with slightly different constants to get the same parameters for the emulator. As we chose $\epsilon' =20\epsilon(r-1)$, we have $\epsilon = O(\frac{\epsilon'}{r})$, which gives $\beta_r = O(\frac{r}{\epsilon'})^{r-1}$. The number of edges is $O(r n^{1+\frac{1}{2^r}})$ in expectation as before. We now bound $\delta_r$.
As $\delta_r = \frac{1}{\epsilon^r} + 2R_r$, and $R_i \leq \frac{2}{\epsilon^{i-1}}$ by Claim \ref{claim_Ri_bound}, we have $\delta_r = O(\frac{1}{\epsilon^r})$. Also, as $\epsilon = O(\frac{\epsilon'}{r})$, we have $\delta_r = O(\frac{r}{\epsilon'})^r$. 

The total stretch for any two vertices is bounded by the stretch in the case $i=r$ which is bounded by $(1+40\epsilon r, 2 \beta_r)$ as discussed above. Since we have $\epsilon' =20\epsilon(r-1)$, we have $40 \epsilon r \leq 4 \epsilon'$ for $r \geq 2$. This bounds the total stretch with $(1+4 \epsilon', O(\frac{r}{\epsilon'})^{r-1})$. Choosing $\epsilon'' = 4 \epsilon'$, gives a stretch of $(1+\epsilon'', O(\frac{r}{\epsilon''})^{r-1})$. 
By Lemma \ref{lemma_time}, the time complexity of the algorithm is $O(\frac{\log^2{\delta_r}}{\epsilon'})$ rounds w.h.p.
Since $\beta_r= O(\frac{r}{\epsilon''})^{r-1}$, $\delta_r=O(\frac{r}{\epsilon''})^r$, and $\epsilon'' = O(\epsilon')$ we have that a complexity of $O(\frac{\log^2{\delta_r}}{\epsilon'})$ is also equal to $O(\frac{\log^2{\beta_r}}{\epsilon''})$. 
\end{proof}

\end{document}